\documentclass{article}
\usepackage{custom_tex}

\begin{document}

\title{Finding Distributions that Differ, with False Discovery Rate Control}

\date{\today}

\author{Yonghoon Lee}
\author{Edgar Dobriban}
\author{Eric Tchetgen Tchetgen\footnote{E-mail addresses: 
\texttt{yhoony31@wharton.upenn.edu},
\texttt{dobriban@wharton.upenn.edu},
\texttt{ett@wharton.upenn.edu}
}}

\affil{Department of Statistics and Data Science, The Wharton School,  \\ University of Pennsylvania}

\maketitle

\begin{abstract}
We consider the problem of 
comparing 
a reference distribution
with several other distributions. 
Given a sample from 
both the reference and the comparison groups, 
we aim to identify the comparison groups whose distributions differ from that of the reference group. 
Viewing this as a multiple testing problem, 
we introduce a methodology that provides exact, distribution-free control of the false discovery rate. To do so, we introduce the concept of \textit{batch conformal $p$-values} and demonstrate that 
they satisfy positive regression dependence across the groups~\citep{benjamini2001control}, 
thereby enabling control of the false discovery rate through the Benjamini-Hochberg procedure. 
The proof of positive regression dependence introduces a novel technique for the inductive construction of rank vectors with almost sure dominance under exchangeability. We evaluate the performance of the proposed procedure through simulations, where, despite being distribution-free, in some cases they show performance comparable to methods with knowledge of the data-generating normal distribution; and further have more power than 
direct approaches based on conformal out-of-distribution detection.
Further, we illustrate our methods on 
a Hepatitis C treatment dataset, where they can identify patient groups with large treatment effects;
and on the Current Population Survey dataset, where they can identify sub-population with long work hours.
\end{abstract}

\tableofcontents
\medskip

\section{Introduction}

We consider the problem of identifying groups whose distributions differ from a given reference distribution. 
Such comparisons arise  frequently, for instance, when evaluating whether a factor such as a treatment impacts a medical outcome compared to a control condition. Formally,
given data sampled from a reference distribution $P$ and comparison distributions $P^{(1)}$, $\ldots$, $P^{(K)}$, we consider testing the $K$ null hypotheses $ H_k : P^{(k)} = P, \quad k = 1, 2, \ldots, K$.
For instance, in evaluating the treatment outcome $T$ in its association 
 with an outcome variable $X$,
we may compare the distributions $P_{X \mid T = t_k}$ of the outcome given treatment level $t_k$ with the distributions $P_{X \mid T = t_0}$ of the outcome given the control condition $t_0$, corresponding to the null
$H_k : P_{X \mid T = t_0} = P_{X \mid T = t_k}$. 

The global null when all null hypotheses $H_k$, $k = 1, 2, \ldots, K$, hold, corresponds to the hypothesis of independence, $H_0 : T \indep X$.
This has been widely studied, with 
classical methods such as permutation tests providing
exact type I error control without requiring explicit distributional assumptions \citep{eden1933validity,fisher1935design,pesarin2001multivariate,ernst2004permutation,pesarin2010permutation,pesarin2012review,good2006permutation,anderson2001permutation,kennedy1995randomization,hemerik2018exact}. 

However, these methods are only informative 
about the global null of whether the factor has any influence.
In practice, researchers often require more nuanced insights, specifically identifying which particular groups (or treatment levels) deviate significantly from the reference. 
This problem can be formulated as a multiple testing problem \citep[e.g.,][etc]{shaffer1995multiple,lehmann2022testing}, in which one simultaneously tests hypotheses $H_k : P_{X \mid T = t_0} = P_{X \mid T = t_k}$. Ensuring reliable inference in this multiple testing setting requires 
control of error measures such as the false discovery rate  \citep{benjamini1995controlling}, ideally under minimal distributional assumptions.

A special case of crucial importance is when we have only \emph{one} comparison distribution: this corresponds to the two-sample testing problem where we test 
$P^{(1)} = P$ \citep[e.g.,][etc]{pitman1937significance,kim2020robust,pesarin2010finite,pesarin2013weak,pesarin2015some,kim2020minimax,lehmann2022testing}.
However, 
addressing multiple comparisons poses unique challenges. In particular, dependence among test statistics arising from 
the shared reference observations 
poses a challenge to the application or construction of valid and powerful multiple testing procedures.

We address this challenge, by proposing a novel distribution-free method for controlling the false discovery rate in the multiple-group comparison setting. 
Our contributions can be summarized as follows.

\paragraph{Batch conformal $p$-values for distribution-free testing.}
We introduce the concept of \textit{batch conformal $p$-values}, an adaptation of batch predictive inference \citep{lee2024batch}, to enable rigorous distribution-free tests for equality of distributions between datasets.
Specifically, the batch conformal $p$-values reject the null hypothesis when the quantiles of the distributions they compare are different.

\paragraph{Multiple testing with exact false discovery rate control via positive regression dependence.}
We establish a 
theoretical result demonstrating that batch conformal $p$-values exhibit positive regression dependence across groups \citep{benjamini2001control}. This property justifies the application of the Benjamini-Hochberg procedure \citep{benjamini1995controlling} for exact control of the false discovery rate  when testing for distributional shifts across multiple groups. The cornerstone of our proof is a novel inductive method for constructing rank vectors that exhibit almost sure dominance under exchangeability.

\paragraph{Empirical validation demonstrating near-oracle performance.}
Through simulations, we demonstrate that our distribution-free approach controls the false discovery rate,
while often achieving power comparable to model-based methods under correctly specified parametric settings.
At the same time, it significantly outperforms direct conformal methods designed for out-of-distribution detection \citep{bates2023testing}. 

\paragraph{Illustrations on empirical datasets.}
We illustrate 
our methods using two datasets:
the HALT-C study on Hepatitis C treatment effects \citep{snow2024haltc}
and the Current Population Survey (CPS) \citep{census_cps}. 
Specifically, we identify patient groups with significant treatment responses in the HALT-C dataset and pinpoint subpopulations exhibiting notably longer working hours in the CPS dataset.

\subsection{Notations}
We denote the set of real numbers by $\R$. 
For a set $\X$, $\X^n$ denotes its $n$-dimensional product. 
For a positive integer $n$, we write $[n]$ to denote the set $\{1,2,\ldots,n\}$ and $\mathcal{S}_n$ to denote the set of permutations $\{\sigma : [n] \rightarrow [n],\,  \sigma \textnormal{ is a bijection}\}$. For a distribution $P$ and $\tau \in (0,1)$, the quantile function $Q_\tau(P) = \inf\{t \in \mathbb{R} : \Pp{X \sim Q}{X \le t} \ge \tau\}$ denotes the $\tau$-th quantile of $P$, where the infimum is defined as infinity if the set is empty. If a finite set or multiset $A \subset \mathbb{R}$ is provided as input to $Q_\tau$, then $Q_\tau(A)$ represents the quantile of the uniform distribution over $A$. The indicator of an event $E$ is denoted by $\One{E}$. For vectors $u, v \in \mathbb{R}^d$, we write $u \preceq v$ if $u_i \leq v_i$ for all $i \in [d]$.

\subsection{Problem setup}
\label{ps}

We consider a setting where our observations are drawn from either a 
\emph{reference distribution}, or 
several \emph{comparison distributions}.
Specifically, 
we
observe a \emph{reference dataset} $(X_1, \ldots, X_n) \subset \X^n$
drawn i.i.d.~from a reference distribution $P$.
Moreover, we observe $K$ independent \emph{comparison datasets} 
$G_1, \ldots, G_K$, where each $G_k = (X_1^{(k)}, X_2^{(k)}, \ldots, X_{n_k}^{(k)}) \subset \X^{n_k}$ is drawn i.i.d.~from a comparison distribution $P^{(k)}$. 
We are interested in selecting groups $k$ whose comparison distribution $P^{(k)}$ is different from the reference distribution. An illustration of the problem setup is given in Figure~\ref{fig:diagram}.

\begin{figure}[!htbp]
    \centering
    \includegraphics[width=0.8\textwidth]{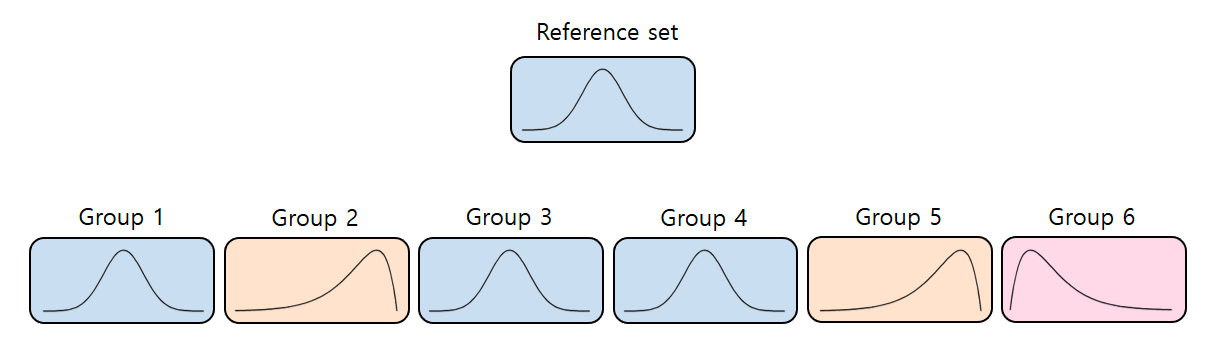}
    \caption{Visual representation of the problem setup. In this example, the objective is to correctly identify the comparison groups---groups 2, 5, and 6---whose distributions differ from that of the reference set, using datapoints from each group.}
    \label{fig:diagram}
\end{figure}

Specifically, we consider the multiple testing problem
with the hypotheses
\begin{align}\label{h}
    H_k : P^{(k)} = P, \quad k = 1, 2, \ldots, K.
\end{align}

As usual in multiple hypothesis testing \citep{lehmann2022testing}, 
our goal is to use the data to select a subset $S\subset \{1, \ldots, K\}$ aiming to capture most of the indices $k$ for which $P^{(k)} \neq P$.
At the same time, we aim to control the number of selected indices $k$ where $P^{(k)} = P$, e.g., where a treatment of interest has no effect.

\subsection{Related work}

The related literature is too vast to summarize here; therefore, we only review the most relevant prior works. 
Testing the equality of distributions has been investigated in various contexts, ranging from parametric methods, such as the two-sample $t$-test~\citet{student1908probable}, to nonparametric approaches, including the Kolmogorov–Smirnov test~\citet{an1933sulla, smirnov1948table}, the energy distance method~\citet{szekely2004testing}, and the maximum mean discrepancy test~\citet{gretton2012kernel}.
Additionally, distribution-free tests, such as the permutation test~\citep{eden1933validity,fisher1935design,pesarin2001multivariate,ernst2004permutation,pesarin2010permutation,pesarin2012review,good2006permutation,anderson2001permutation,kennedy1995randomization,hemerik2018exact} 
and the rank-sum test~\citet{mann1947test}, have been widely studied. For more recent developments, see~\citet{rosenbaum2005exact} and~\citet{biswas2014distribution}.

While the aforementioned studies primarily focus on comparing two distributions, our work addresses the multiple hypothesis testing problem, wherein multiple comparison distributions are compared against a reference distribution. 
Similar problems
have been explored in several works, starting with the Dunnett procedure \citep{dunnett1955multiple} for normally distributed data. \citet{hothorn2020comparisons} provided a simulation study comparing the Dunnet $t$-test with closed testing procedures. 
In the distribution-free inference literature, such many-to-one comparisons have been extensively explored in the context of outlier/novelty detection. \citet{bates2023testing} presented a strategy for using conformal $p$-values to detect outliers while controlling the false discovery rate in a distribution-free setting. 
Their method was extended in \citet{marandon2024adaptive} and \citet{magnani2024collective}. Additionally, \citet{lee2025full} proposed a methodology based on full-conformal e-values \citep{vovk2021values}, enabling more efficient use of the data.

Controlling the false discovery rate \citet{benjamini1995controlling} in multiple testing problems has been a topic of significant interest, starting with the Benjamini-Hochberg  procedure introduced by \citet{benjamini1995controlling}. \citet{benjamini2006adaptive} propose an extension of the procedure with tighter false discovery rate control. Several subsequent works have studied  conditions beyond independence on the input $p$-values under which the Benjamini-Hochberg procedure achieves valid false discovery rate control.
\citet{benjamini2001control} showed that this holds under the more general condition of 
positive regression dependence  that they introduced.
The theoretical properties under positive regression dependence
were further studied in~\citet{sarkar2002some}. \citet{blanchard2008two} studied conditions for false discovery rate control with extensions of the Benjamini-Hochberg procedure.
See \cite{benjamini2010discovering} for a review.

\section{Methodology and theoretical results}

We now introduce our methods identifying distribution shifts across multiple groups. 
The basic approach involves applying the Benjamini-Hochberg  procedure \cite{benjamini1995controlling} after constructing valid $p$-values.
We first provide a brief overview of the Benjamini-Hochberg procedure, 
examine a few simple methods for the multiple comparison problem, and discuss their limitations.
Then, we introduce our procedure and prove that it controls the false discovery rate.
Finally, we discuss how our 
methods specialize to the simpler but important case of comparing \emph{two} distributions.

\subsection{Review of the Benjamini-Hochberg procedure}
\label{bhrev}

In the setting of multiple hypothesis testing, 
the Benjamini-Hochberg  procedure~\citep{benjamini1995controlling} is a widely-used methodology.
Given a target level $\alpha \in (0,1)$ and $p$-values $p_1, \dots, p_M$ for hypotheses $H_1, \dots, H_M$, the Benjamini-Hochberg procedure computes $k^* = \max\{k : p_{(k)} \leq \frac{k}{M} \alpha\}$, where $p_{(k)}$ denotes the $k$-th order statistic of the $p$-values. 
It then rejects all null hypotheses for which the corresponding $p$-value is less than or equal to $p_{(k^*)}$.

The false discovery rate is defined as
$\EE{V/\max(R,1)}$, where $V$ is the number of false discoveries, i.e., the number of rejected hypotheses that are true or null, and $R$ is the number of rejected hypotheses. The Benjamini-Hochberg procedure is known to control the false discovery rate at the predefined level $\alpha$, when the $p$-values $p_1, \dots, p_M$ are independent \citep{benjamini1995controlling}, or when they satisfy the \textit{positive regression dependence on a subset} property on the set of true null hypotheses \citep{benjamini2001control}. 

To define this property, we first introduce the following definitions: For a positive integer $d$ and two vectors $v = (v_1, \ldots, v_d)^\top$ and $w = (w_1, \ldots, w_d)^\top$ in $\mathbb{R}^d$, we write $v \preceq w$ if $v_j \leq w_j$ for all $j \in [d]$. A set $A \subset \mathbb{R}^d$ is called increasing if, for any $x \in A$, $x \preceq y$ implies that $y \in A$.

\begin{definition}[\citet{benjamini2001control}]\label{def:prds}
For a set $I\subset \{1, \ldots, K\}$,
    a random vector $X = (X_1,X_2$, $\ldots$, $X_K)$ is positive regression dependent on $I$ if for all $k \in I$, the conditional probability $x\mapsto \PPst{X \in A}{X_k = x}$ is nondecreasing on its domain as a function of  $x \in \R$ for any increasing set $A \subset \R^{K}$.
\end{definition}

\citet{benjamini2001control} shows that if the $p$-value vector $(p_1,\ldots,p_M)^\top$ is positive regression dependent on the set of nulls $I_0 = \{k : H_k \textnormal{ is true}\}$, then the Benjamini-Hochberg procedure controls the false discovery rate at the desired level $\alpha$, i.e., $\EE{V/\max(R,1)} \le \alpha$.

\subsection{Existing methods}

In this section, we discuss several existing---i.e., baseline---approaches that one may consider 
for constructing $p$-values for our selection problem, and illustrate their limitations.

\subsubsection{Partitioning the reference dataset}

As discussed in the previous section, one option for controlling the false discovery rate is to construct independent $p$-values for $H_1, \dots, H_K$, and then apply the Benjamini-Hochberg procedure.
The methods discussed in Section~\ref{sec:two_sample} can be used for constructing a $p$-value for an individual hypothesis. 
However, datapoints should be used \emph{only once}, 
in order to ensure the independence of the $p$-values. 
To achieve this, 
one could split the 
reference dataset into $K$ 
subsets,
and use each subset to construct a $p$-value for each hypothesis $H_k$, $k\in [K]$, from \eqref{h}. 
For example, if $n = Kq + r$ for some positive integer $q$ and $0 \leq r < q$, then one can use the dataset $\{X_{(k-1)q+1}, X_{(k-1)q+2}, \dots, X_{(k-1)q+q-1}\}$ to construct a permutation test $p$-value for $H_k$.

However, this approach
becomes problematic when the reference sample size is small.
In that case, splitting the reference data into $K$ subsets leads to very small datasets what will lead to inaccurate $p$-values with low power.
To handle this setting, 
it is desired to develop methods that permit reusing the reference data.

\subsubsection{Testing with conformal p-values after subsampling}\label{sec:subsampling}

An approach that allows reusing
the reference data is based on 
conformal testing
for  outliers, which has appeared in \cite{bates2023testing} as part of a broader effort to construct calibration-set conditional $p$-values for outlier detection; see also \cite{mary2022semi}.

To introduce this, 
let $s : \X \rightarrow \R$ be a \textit{score} function that maps each observation to a real value and is constructed independently of the data.
In practice, it is standard in related problems in conformal prediction to split the data into two subsets, using one to construct the score function and the other as the reference dataset \citep{vovk2005algorithmic}. 
We choose a score such that 
larger scores are more likely to occur under $P^{(k)}$  when $P\neq P^{(k)}$. 
Let $S_i$, $i\in[n]$, and $S_i^{(k)}$, $k\in[K]$, $i\in[n_k]$,
denote the \emph{reference scores} and the \emph{comparison scores} of the $k$-th comparison group, respectively. 

To use 
conformal testing
for  outliers \citep{bates2023testing},
we randomly choose one datapoint from each $G_k$ and
apply the Benjamini-Hochberg procedure to the following 
$p$-values:
\[p_k = \frac{\sum_{i=1}^n \mathbf{1}\{S_{i_k^*}^{(k)} \leq S_i\} + 1}{n+1}, \quad \text{where } i_k^* \sim \text{Unif}(\{1,2,\ldots, n_k\}).\]
\cite{bates2023testing} consider detecting outliers, so they consider one datapoint from each distribution, which in our context implies that their method requires 
subsampling.
\citet{bates2023testing} shows that the
above \emph{subsampling conformal $p$-values} 
satisfy positive regression dependence, so that the Benjamini-Hochberg procedure ensures false discovery rate control.

However, 
this procedure discards most of the information from the comparison data, relying only on one randomly selected data point from each group.
In our experiments, we show that this can 
lead to low power compared to our proposed method.
Moreover, the above line of work also provides another interpretation and possible application of our work to \emph{batch outlier detection}, where we have batched observations and use our method to detect outlier batches.

\subsubsection{Applying Benjamini-Yekutieli procedure}\label{sec:BY}

Alternatively, one may consider applying the Benjamini–Yekutieli procedure~\citep{benjamini2001control} instead of the Benjamini–Hochberg procedure to any two-sample $p$-values. The Benjamini–Yekutieli procedure applies the Benjamini–Hochberg procedure at a more conservative level $\alpha / \sum_{i=1}^K (1/i)$ and provides valid false discovery control under arbitrarily dependent $p$-values. However, this procedure typically produces quite conservative results; therefore, we explore a way to construct $p$-values that allow the use of the Benjamini–Hochberg procedure.

\subsection{Proposed method: testing with batch conformal p-values}

We now present our main procedure, which allows the repeated use of reference data. 
Consider a score function $s : \X \rightarrow \R$ constructed in advance, independently of the data used for inference.
See \cite{vovk2005algorithmic,angelopoulos2021gentle} for standard examples.
For example, one can construct an estimated mean function $\hat{\mu}(\cdot)$ using related data,
and then work with the residual score $s: (x,y) \mapsto |y-\hat{\mu}(x)|$. If we are interested to compare the outcomes themselves, we can choose 
$s: (x,y) \mapsto y$. 

Further, select integers $\eta_k \in [n_k]$ for each $k \in [K]$, such that interest centers on the $\eta_k/n_k$th quantile of the distribution of the scores.
Then, we define the \textit{batch conformal $p$-value} 
for $H_k$ as\footnote{Here, for non-negative integers $a\le b$, $\binom{b}{a} = b!/(a!(b-a)!)$ denotes the binomial coefficient, where $x!=x\cdot (x-1) \cdot \ldots\cdot  1$ is the factorial of a non-negative integer $x$.}
\begin{equation}\label{eqn:p_val_k}
    p_k = \sum_{i=1}^n \frac{\binom{i+\eta_k-2}{\eta_k-1}\binom{n+n_k-i-\eta_k+1}{n_k-\eta_k}}{\binom{n+n_k}{n_k}} \cdot \One{S_{(\eta_k)}^{(k)} \leq S_{(i)}} + \frac{\binom{n+\eta_k-1}{\eta_k-1}}{\binom{n+n_k}{n_k}}.
\end{equation}
The $i$th term in this expression computes the probability that the \(\eta_k\)-th smallest test score among a random size-\(n_k\) subset drawn from a combined population of \(n + n_k\) scores (consisting of both calibration and test points) is \emph{at least as large} as the \(i\)-th smallest calibration score \citep[e.g.,][p. 243]{wilks1962mathematical}. 
The first term sums over all \( i \in [n] \), weighting by the probability that the \(\eta_k\)-th order statistic lands at position \( i \), and then checking if the observed test statistic \( S_{(\eta_k)}^{(k)} \le S_{(i)} \).
The final term corresponds to the probability that the \(\eta_k\)-th smallest element in the random sample is larger than all \( S_{(i)} \).
Intuitively, testing with the batch conformal $p$-value rejects the null $H_k$ if the 
quantile of the comparison scores $S_{(\eta_k)}^{(k)}$---which can be viewed as the test statistic for $H_k$---is unusually large compared to the reference scores. 

For example, if one chooses to use the median as the test statistic, one can set, e.g., $\eta_k = \lceil 0.5 \cdot n_k \rceil$.
We remark that statistical inference for quantiles has been widely studied \citep{kosorok1999two,sgouropoulos2015matching,yang2017smoothed,lou2025high}. However, we are not aware of our specific method being widely studied.
Specifically, as we will show, our method is valid in finite samples without distributional assumptions, while being powerful when the quantiles of the distributions are different. Thus, it
has different properties from other methods such as those in \citep{kosorok1999two} which are valid 
only asymptotically and require certain distributional assumptions.

When the comparison dataset has \emph{one datapoint}, i.e., $n_k=\eta_k=1$, the batch conformal $p$-value coincides with the standard conformal $p$-value, given by $\frac{\sum_{i=1}^n \One{S_{1}^{(k)} \leq S_i}+1}{n+1}$. 
We will discuss the choice of the score $s$ later. Further, we will discuss the use of batch conformal $p$-values for the special case of two-sample testing in Section \ref{sec:two_sample}.

The idea of batch conformal $p$-values is motivated by the predictive inference method for a quantile of multiple test points developed in~\citet{lee2024batch}, leveraging classical results about the distribution of quantiles in finite populations, see e.g., \cite{wilks1962mathematical}. 
It follows from existing results that the
formulas in \eqref{eqn:p_val_k}
lead to marginally valid $p$-values when  the scores $S_1, \ldots, S_{n}, S_{1}^{(k)}, \ldots, S_{n_k}^{(k)}$ are exchangeable.
For completeness, we provide a self-contained proof 
in Appendix~\ref{pfprop:pval}.
As a consequence, 
$p_1,p_2,\ldots,p_K$ forms a marginally valid
vector of $p$-values.

Next, 
let $I_0 = \{k \in [K] : \text{$H_k$ is true}\}$ denote the set of group indices for which the null hypotheses from \eqref{h} hold, and let $K_0 = |I_0|$ denote the number of true nulls. 
We now prove our main result, 
that applying the Benjamini-Hochberg procedure with the batch conformal $p$-values ensures valid false discovery rate control.

\begin{theorem}[Main result: False discovery rate control]\label{thm:prds}
    Let the reference scores $(S_i)_{i \in [n]}$ and the comparison scores 
    $\smash{(S_j^{(k)})_{j \in [n_k],k \in [K]}}$ be almost surely all distinct, 
    and the score groups $(S_i)_{i \in [n]}$ and
    $\smash{\{(S_j^{(k)})_{j \in [n_k]} : k \in [K]\}}$ be independent.
    Then the batch conformal $p$-values $p_1,p_2,\ldots,p_K$ defined in~\eqref{eqn:p_val_k} satisfy the positive regression dependence condition (Definition \ref{def:prds}) on the set $I_0$ of true nulls. 
    Consequently, the Benjamini-Hochberg procedure at level $\alpha$ (Section \ref{bhrev} and Algorithm \ref{alg:batch}),
    applied with $(p_k)_{k \in [K]}$, controls the false discovery rate at level $K_0\alpha/K \le \alpha$.
\end{theorem} 

\begin{remark}
    The validity of batch conformal $p$-values implies that we can in fact test the following weaker null hypotheses:
    \[H_k : \text{ The $n+n_k$ scores } (S_i)_{i \in [n]} \text{ and } (S_j^{(k)})_{j \in [n_k]} \text{ are exchangeable}, \qquad k \in [K].\]
    The proof of Theorem~\ref{thm:prds} ensures that the Benjamini-Hochberg procedure applied to the batch conformal $p$-values controls the false discovery rate for these nulls, 
    if the comparison score groups $\{(S_j^{(k)})_{j \in [n_k]} : k \in [K]\}$ are independent conditional on the reference score group.

Note that a statement $H_k^0$, which states that the data points $(X_i, Y_i)_{i \in [n]}$ and $(X_j^{(k)}, Y_j^{(k)})_{j \in [n_k]}$ are exchangeable, implies the above $H_k$, as long as the score function $s$ is either fixed or constructed using a data split. Therefore, the proposed procedure also provides FDR control for the nulls $H_k^0$, $k \in [K]$.
\end{remark}

\begin{remark}
The positive regression dependence property established in Theorem~\ref{thm:prds} also has implications for the global null testing problem. Specifically, consider the problem of testing whether there exists \textit{at least one} comparison dataset whose distribution differs from that of the reference data. A straightforward approach to control the type I error would be the Bonferroni test, which accommodates arbitrary dependence among p-values and can therefore be applied with any existing two-sample p-values. However, the Bonferroni test is well known to be conservative. 

In contrast, Simes test offers a less conservative alternative but requires certain conditions on the p-values. \citet{sarkar2008simes} shows that Simes test remains valid under positive regression dependence, implying that, by Theorem~\ref{thm:prds}, the batch conformal p-values can also be used within Simes test while ensuring valid type I error control.
\end{remark}

\begin{algorithm}
\caption{Distribution-shift detection with batch conformal $p$-values}
\label{alg:batch}
{\bf Input:} Reference data $\mathcal{D}_n= \{X_1,X_2,\ldots,X_n\}$. Comparison sets $(X_1^{(k)}, X_2^{(k)}, \ldots, X_{n_k}^{(k)})_{1 \leq k \leq K}$. Score function $s: \X \rightarrow \R$. Target level $\alpha$. Rank vector $(\eta_1,\ldots,\eta_K) \in [n_1] \times \ldots \times [n_K]$.

{\bf Step 1:} Compute the scores $S_i = s(X_i)$ for $1 \leq i \leq n$ and $S_i^{(k)}= s(X_i^{(k)})$ for $1 \leq k \leq K$ and $1 \leq i \leq n_k$.

{\bf Step 2:} Compute the batch conformal $p$-values $p_1,\ldots,p_K$ as
\[p_k = \sum_{i=1}^n \frac{\binom{i+\eta_k-2}{\eta_k-1}\binom{n+n_k-i-\eta_k+1}{n_k-\eta_k}}{\binom{n+n_k}{n_k}} \cdot \One{S_{(\eta_k)}^{(k)} \leq S_{(i)}} + \frac{\binom{n+\eta_k-1}{\eta_k-1}}{\binom{n+n_k}{n_k}},\; k \in [K],\]
and sort them as $p_{(1)} \leq \ldots \leq p_{(K)}$.

{\bf Step 3:} Run the Benjamini-Hochberg procedure to find $k^* = \max\{k : p_{(k)} \leq \frac{k}{K}\alpha\}.$

{\bf Return:} Rejection set $\mathcal{R} = \{k \in [K] : p_k \leq p_{(k^*)}\}$ ($\mathcal{R} = \emptyset$ if $\min_{1 \leq k \leq K} \frac{K}{k}p_{(k)} > \alpha$).
\end{algorithm}

The overall procedure is described in Algorithm~\ref{alg:batch}, and the proof of Theorem~\ref{thm:prds} is deferred to Appendix~\ref{pfthm:prds}. 
The proof of the theorem introduces a novel technique, which differs from proof techniques used in related results such as in \cite{bates2023testing}. 
Our approach introduces a technique for constructing random variables $(\tilde{R}_1, \ldots, \tilde{R}_m)$ and $(\hat{R}_1, \ldots, \hat{R}_m)$ inductively over $m$, ensuring that they have the same distribution as the ranks of the test scores conditional on two different values of $p_k$, while also ensuring that $(\tilde{R}_1, \ldots, \tilde{R}_m) \preceq (\hat{R}_1, \ldots, \hat{R}_m)$ holds almost surely. 
To enable this construction, we show that the conditional distribution of the ranks $R_{t-1}$ given $R_t$ and the set of scores is stochastically increasing in $R_t$ (Lemma~\ref{lem:ranks_2}).
Then Strassen's theorem~\citep{strassen1965existence} is applied repeatedly to construct appropriate random variables
$\tilde{R}_{t-1} \leq \hat{R}_{t-1}$ having the same conditional distributions as  $R_{t-1}$ given the set of scores  
and two appropriate values of $R_t$. 
This construction may have broader use cases beyond the scope of our paper.

\subsubsection{Choice of the score function}

By Theorem~\ref{thm:prds}, Algorithm~\ref{alg:batch} ensures valid false discovery rate control for any score function
constructed independently of the reference and comparison data. However, the power of the test depends on this choice. 
Here, we discuss several possible choices for the score function across different scenarios.

\paragraph{One-dimensional outcome.}
Let the observed outcomes  be scalars: 
$Y_1,\cdots,Y_n \in \R$ for the reference set, and $Y_1^{(k)},\cdots,Y_{n_k}^{(k)} \in \R$, $k \in [K]$ for each comparison group $k \in [K]$.
To detect groups with distribution shift, 
one can set the score function as simply $s(y) = y$ or $s(y) = -y$ for all $y$, depending on the direction of the shift of interest. 
For two-sided detection, one direct option is to split the reference data into two sets, use one set to construct an estimate $\hat{\mu}$ of the ``center" (e.g., mean or median), and then set the score function as $s:y \mapsto |y - \hat{\mu}|$ to run the procedure with the second split.

\paragraph{One-dimensional outcome with side information.}

Consider feature-label pairs
$(X_1, Y_1)$, $ \ldots$, $ (X_n, Y_n)$ $ \in \mathbb{R}^d \times \mathbb{R}$ for the reference group, and similarly for the comparison groups.
In this setting, standard nonconformity scores \citep{vovk2005algorithmic} can be used---for example, one can apply data splitting to the reference set and use one split to construct scores such as \citep{romano2019conformalized}
\[s:(x, y) \mapsto |y - \hat{\mu}(x)| \text{ or } s:(x, y) \mapsto \max\{\hat{q}_{\alpha/2}(x) - y, y - \hat{q}_{1-\alpha/2}(x)\},\]
where $\hat{\mu}(\cdot)$ is an estimator of the conditional mean $\mathbb{E}[Y \mid X]$ and $\hat{q}_{\alpha/2}(\cdot)$, $\hat{q}_{1-\alpha/2}(\cdot)$ are estimators of $\alpha/2,1-\alpha/2$ conditional quantiles. 
Being error measures for an estimator 
fitted to be small under the null (reference) distribution,
we expect these scores to be larger when there is a shift.

\paragraph{Multivariate outcome variable.}

Now consider features with
multivariate outcomes: $(X_1, Y_1)$, $ \ldots$, $(X_n, Y_n) $ $\in \mathbb{R}^d \times \mathbb{R}^p$, where each $Y_i = (Y_{i1}, \ldots, Y_{ip})^\top \in \mathbb{R}^p$---and similarly for the comparison groups.
One option is to use the training data (one split of the reference dataset) to construct $\hat{\mu}_1, \dots, \hat{\mu}_p$ such that $\hat{\mu}_j(X)$ predicts the $j$-th component of the outcome $Y$, and then combine them. 
For example, one can set the score $s: \R^d \times \R^p \rightarrow \R$ as 
\(s:(x, y) \mapsto \sum_{j=1}^p |y_j - \hat{\mu}_j(x)|^2 = \|\hat{r}(x, y)\|^2,\)
where $\hat{r}(x, y) = \big(|y_j - \hat{\mu}_j(x)|\big)_{1 \leq j \leq p}$ denotes the residual vector.

To address the potentially different scales and dependencies between the components, one can alternatively set the score as 
\(s:(x, y) \mapsto \hat{r}(x, y)^\top \hat{S}^{-1} \cdot \hat{r}(x, y),\)
where $\hat{S}$ denotes the sample covariance matrix of the residual vectors evaluated on the training data. To better capture the dependence structure within the outcome vector, one can also consider constructing the residual vector in the form of
\[\hat{r}:(x, y) \mapsto \big(y_1 - \hat{\mu}_1(x), \, y_2 - \hat{\mu}_2(x, y_1), \, y_3 - \hat{\mu}_3(x, y_1, y_2), \, \dots, \, y_p - \hat{\mu}_p(x, y_1, \dots, y_{p-1})\big)^\top,\]
with $\hat{\mu}_j : \R^d \times \R^{j-1} \rightarrow \R$ constructed accordingly for each $j \in [p]$.

\begin{remark}[Choice of the hyperparameters $\eta_k$]
The proposed procedure involves hyperparameters $(\eta_k)_{k \in [K]}$, which need to be selected by the practitioner. A natural default choice is the median, i.e., setting $\eta_k = \lceil n_k/2 \rceil$, although other choices can be made depending on the target of interest. For example, if $X$ is one-dimensional and $s(x) = x$, and the goal is to detect shifts in the tails of the outcome distribution, a larger value such as $\eta_k = \lceil 0.9 \cdot n_k \rceil$ can be used.

Alternatively, for more complex settings---e.g., with multivariate $X$ and non-identity scores---, one may consider tuning the hyperparameter based on the data to maximize power. For example, for each $k \in [K]$, let $G_k^0 = (X_1^{(k)}, \ldots, X_{m_k}^{(k)})$ denote a split of the comparison data $G_k$ with $m_k < n_k$, and let $G^0 = (X_1, \ldots, X_m)$ be a split of the reference data (which can be the same split used for constructing the score function). Then, for values $q$ in a grid, e.g., $(0.1, 0.2, \ldots, 0.9)$, we compute the (one-sided) difference between the $q$-sample quantiles of the score sets $(s(X_1^{(k)}), \ldots, s(X_{m_k}^{(k)}))$ and $(s(X_1), \ldots, s(X_m))$, and select the value $q_k$ that maximizes this difference. The remaining split $G_k^1 = (X_{m_k+1}^{(k)}, \ldots, X_{n_k}^{(k)})$ is then used for inference, with $\eta_k = \lceil q_k \cdot (n_k - m_k) \rceil$.
\end{remark}

\subsection{Special case: distribution-free two-sample test}\label{sec:two_sample}

Here, we briefly discuss a special case of the setting introduced in Section \ref{ps}: 
having only one comparison dataset, which corresponds to two-sample testing. 
Suppose we have two datasets,
the reference dataset
$X_1, \ldots, X_n \iidsim P$ and 
the comparison dataset
$X_{n+1}, \ldots, X_{n+m} \iidsim Q$, and aim to test the null hypothesis 
\begin{equation}\label{eqn:null}
    H_0 : P = Q.
\end{equation}
This is clearly a special case of the problem introduced in \Cref{ps}, where $K=1$, $P^{(1)}=Q$, and $n_1=m$.
In this two-sample problem, there is no need for multiple testing.
Instead,
the problem reduces to
constructing a test $\phi : \X^n \times \X^m \rightarrow \{0,1\}$ with finite-sample type I error control  
\[
\Ep{H_0}{\phi((X_1,\ldots,X_n),(X_{n+1}, \ldots, X_{n+m}))} = \Ep{X_1,\ldots,X_{n+m} \iidsim P}{\phi((X_1,\ldots,X_n),(X_{n+1}, \ldots, X_{n+m}))} \leq \alpha
\]  
for a predefined level $\alpha \in (0,1)$, under any distribution $P$. 
Well-known methods that achieve this goal include the permutation test~\citep{eden1933validity,fisher1935design,dwass1957modified} and the rank-sum test~\citet{mann1947test}. 
We briefly review these methods in Section \ref{2s-det}:
the permutation test can have a high computational cost or requires randomization to avoid this, 
while the rank-sum test is usually used with its asymptotic normal approximation.

\subsubsection{Two-sample testing with a batch conformal p-value}\label{sec:batch_conformal}

Here, we discuss an alternative method for constructing a $p$-value that is valid in a distribution-free sense, fast to compute for large sample sizes, and also non-randomized. 
With the score function $s : \X \rightarrow \R$
chosen such that 
larger scores are more likely to occur under $Q$  when $P\neq Q$,
define $S_i = s(X_i)$ for $i \in [n+m]$. Now, let $S_{(i)}$ denote the $i$-th order statistic of
the \emph{reference scores}
$S_1,\ldots,S_n$, and $S_{(\eta)}^\text{test}$ the $\eta$-th order statistic of
the \emph{comparison scores}
$S_{n+1}, \ldots, S_{n+m}$, for a predetermined $\eta \in [m]$. 
In this case, the batch conformal $p$-value 
for testing $P=Q$ from \eqref{eqn:p_val_k} becomes
\begin{equation}\label{eqn:p_val}
    p = \sum_{i=1}^n \frac{\binom{i+\eta-2}{\eta-1}\binom{n+m-i-\eta+1}{m-\eta}}{\binom{n+m}{m}} \cdot \One{S_{(\eta)}^\text{test} \leq S_{(i)}} + \frac{\binom{n+\eta-1}{\eta-1}}{\binom{n+m}{m}}.
\end{equation}
As a consequence of the validity of batch conformal $p$-values,
if the scores $S_1, \ldots, S_{n+m}$ are exchangeable, then the statistic $p$ in~\eqref{eqn:p_val} is a valid $p$-value for the null hypothesis $H_0$ in~\eqref{eqn:null}.


\begin{remark}
If the score distribution has point masses, 
we apply uniform tie-breaking to determine the order of scores. Another standard approach is to
add a small amount of noise to the score \citep{kuchibhotla2020exchangeability},
so that the resulting adjusted score has a nonatomic distribution. 
For example, one can 
construct the $p$-value using the scores $\tilde{S}_i = S_i + \epsilon_i$, where $\epsilon_i \iidsim \mathcal{N}(0,10^{-10})$.
While the validity of batch conformal $p$-values holds with the uniform tie-breaking strategy without introducing additional noise,  
the tie-breaking-with-noise strategy can be useful, 
particularly with multiple comparison groups---the main focus of this work---followed by multiple $p$-values reusing the same reference dataset, where ``uniform tie-breaking" may result in a lack of ``consistency''  across groups, thereby requiring the almost-sure distinctness condition in Theorem~\ref{thm:prds}.
\end{remark}

By the validity of batch conformal $p$-values, the test $\phi = \One{p \leq \alpha}$ controls the type I error at level $\alpha$, for any $\alpha \in (0,1)$. 
Intuitively, this test rejects the null if the $\eta$-sample quantile of the test scores is unusually large compared to the reference scores. In Appendix~\ref{sec:multi_quantile}, we present a generalized version of the $p$-value from~\eqref{eqn:p_val} that depends on multiple quantiles.

While our main focus is on the setting of multiple comparison groups, the batch conformal $p$-value 
can have benefits over well-known methods in some scenarios.
First, the goal of our method is to detect changes in the \emph{quantiles} of the distributions, and thus can be used to detect changes in the tails (e.g., the new treatment works better for the top 10\% patients with the highest blood pressure).
In contrast, the popular rank-sum test aims to detect changes in the "bulk" of the distribution.
This allows our method to detect a distinct and complementary kind of effects; see the illustration in Section \ref{ts-de}.

Moreover, our method also has an advantage compared to permutation tests. While one can use these with the same quantile test statistic as our method (and thus in principle they can detect the same effects), a key limitation is that permutation tests can be computationally heavy, requiring the re-computation of test statistics over many permutations.
While one can randomly sample permutations \citep{lehmann2022testing}, one may still need to sample a large number to get significant results (at least $\lceil 1/\alpha\rceil-1$ to have the possibility of a $p$-value less than or equal to $\alpha$).
This is especially severe in situations where there is \emph{multiplicity}, i.e., when multiple tests are performed, and where higher levels of significance are required due to multiplicity adjustments.
For instance, in genomics, permutation tests are performed often, but are known to be expensive \citep[see e.g.,][etc]{stranger2007population,salojarvi2017genome,maura2022efficient}.
Our methods could lead to significant computational savings in such situations.

\section{Simulations}

We present simulation results to illustrate the performance of the proposed procedure. 
We begin with a univariate setting with normal variables for illustration, 
followed by results in a more complex multivariate setting.\footnote{Code to reproduce the experiments in this section is available at \url{https://github.com/yhoon31/batch_conformal}.}

\subsection{Simple case: detecting shifts in a one-dimensional outcome}

We first illustrate the performance of Algorithm~\ref{alg:batch} in a simple setting with a one-dimensional outcome and no covariates. We generate the reference dataset of size $n=100$ from the $\mathcal{N}(0,3^2)$ distribution. Next, we generate multiple comparison datasets, with their sizes $(n_k)_{1 \leq k \leq K}$ drawn uniformly in advance between 30 and 50.
The simulation is run under two null proportions, 0.5 and 0.7, and three numbers of comparison groups: 20, 50, and 200. 
For the non-null groups, 
the data are drawn from $\mathcal{N}(\delta, 3^2)$, using signal strengths $\delta = 1, 2, 3$. The results are shown in Figure~\ref{fig:one_dim}, illustrating that the proposed procedure tightly controls the false discovery rate at the theoretically proven level, supporting the conclusion of \Cref{thm:prds}.

\begin{figure}[!htbp]\
    \centering
    \includegraphics[width=0.8\textwidth]{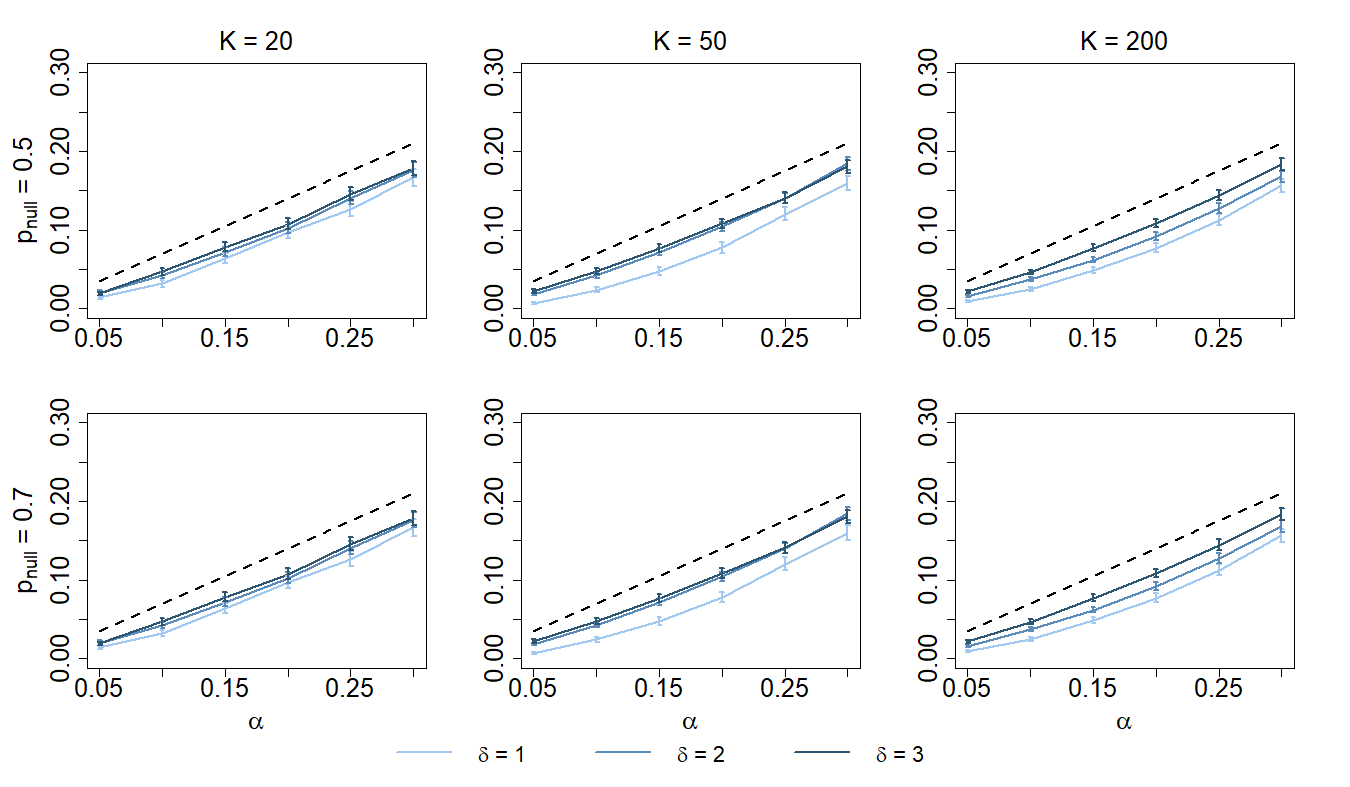}
    \caption{False discovery rate of Algorithm~\ref{alg:batch} across various signal strengths and numbers of comparison groups, under a true-null proportion of $p_\text{null} = 0.5$ and $p_\text{null} = 0.7$. The dotted line corresponds to the theoretical bound $p_\text{null} \cdot \alpha$ from \Cref{thm:prds}.}
    \label{fig:one_dim}
\end{figure}

Next, we compare the performance of Algorithm~\ref{alg:batch} with an ``oracle procedure" that applies the Benjamini-Hochberg procedure to $p$-values obtained from two-sample $z$-tests, given as
\[p_k = \Phi\left(
(\bar{X}_{\mathrm{ref}} - \bar{X}_k)/
\left(\sigma{\sqrt{\frac{1}{n} + \frac{1}{n_k}}}\right)\right),\quad k=1,2,\ldots,K,\]
where $\bar{X}_{\mathrm{ref}}$ and $\bar{X}_k$ represent the sample means of the reference dataset and the $k$-th comparison dataset, respectively, and $\Phi$ denotes the cumulative distribution function of a standard normal random variable. 
Here, the variance is assumed to be known, with $\sigma = 3$ in our examples. 
These $p$-values
satisfy positive regression dependence (see Appendix~\ref{sec:z_prds}), 
and thus the Benjamini-Hochberg procedure guarantees false discovery rate control. 
We generate the data as before under $\delta=1,2,3$ and $K=50$. The results are shown in Figure~\ref{fig:sim_z}, illustrating that 
\emph{our proposed procedure 
attains power  comparable to the oracle procedure}.

\begin{figure}[!htbp]
    \centering
    \includegraphics[width=0.8\textwidth]{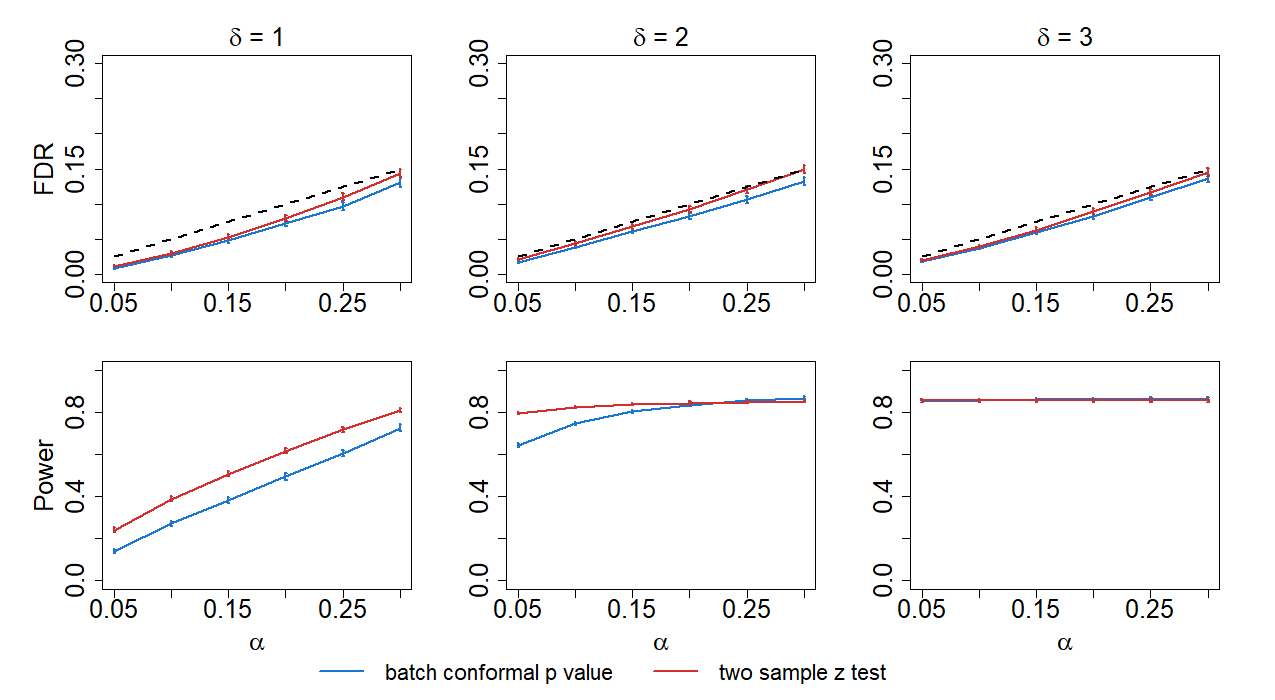}
    \caption{False discovery rate and power of Algorithm~\ref{alg:batch} and the procedure with the $p$-values from the two sample $z$-test, across various signal strengths, under a true-null proportion of $0.5$. The dotted line corresponds to the theoretical bound $0.5 \alpha$ from \Cref{thm:prds}.}
    \label{fig:sim_z}
\end{figure}

We also provide a result for the case where the true distribution is not normal, in which approaches based on the normal model assumption might not be successful. We repeat the same simulation with the following null and alternative distributions:
\[\textnormal{Null: } X \sim 0.5 \cdot \textnormal{Cauchy}(0,1) + 0.5 \cdot \textnormal{Unif}[-1,1], \quad \textnormal{Alternative: } X \sim \delta + 0.5 \cdot \textnormal{Cauchy}(0,1) + 0.5 \cdot \textnormal{Unif}[-1,1],\]
where we set $\delta = 1$. 
We compare the results 
of our Algorithm~\ref{alg:batch} with the results of the two-sample $t$-test $p$-values. Figure~\ref{fig:sim_non_normal} summarizes the results for true null proportions of $0.3$, $0.5$, and $0.7$. 
The false discovery rate of the procedure based on $z$-test $p$-values occasionally exceeds the bound $(\text{true null proportion}) \cdot \alpha$ but remains controlled at level $\alpha$.
However, its power is significantly lower than that of the distribution-free Algorithm~\ref{alg:batch}.

\begin{figure}[!htbp]
    \centering
    \includegraphics[width=0.8\textwidth]{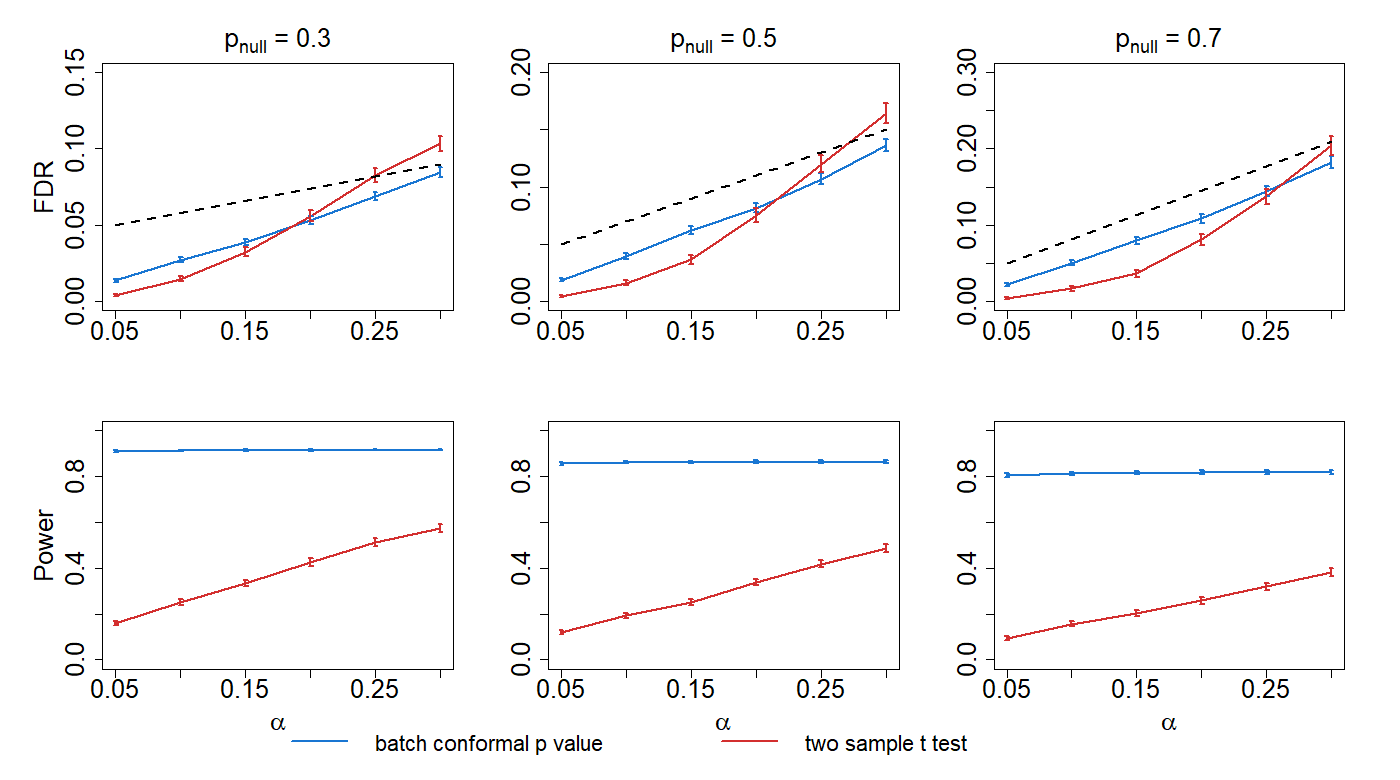}
    \caption{False discovery rate and power of Algorithm~\ref{alg:batch} and the procedure with the $p$-values from the two sample $t$-test, under a non-normal null distribution, across true-null proportions of $0.3, 0.5$, and $0.7$. The dotted line corresponds to the theoretical bound from \Cref{thm:prds}.}
    \label{fig:sim_non_normal}
\end{figure}

Next, we compare the performance of the proposed procedure to the method discussed in Section~\ref{sec:subsampling}, which uses the conformal $p$-value constructed by subsampling. Figure~\ref{fig:sim_subsampling} illustrates that both methods control the false discovery rate at the desired level, but the subsampling-based approach often leads to low power. Intuitively, this is likely because using only a single sample in the non-null case does not provide sufficient evidence to conclude that the sampling distribution differs from that of the reference set.

\begin{figure}[!htbp]
    \centering
    \includegraphics[width=0.8\textwidth]{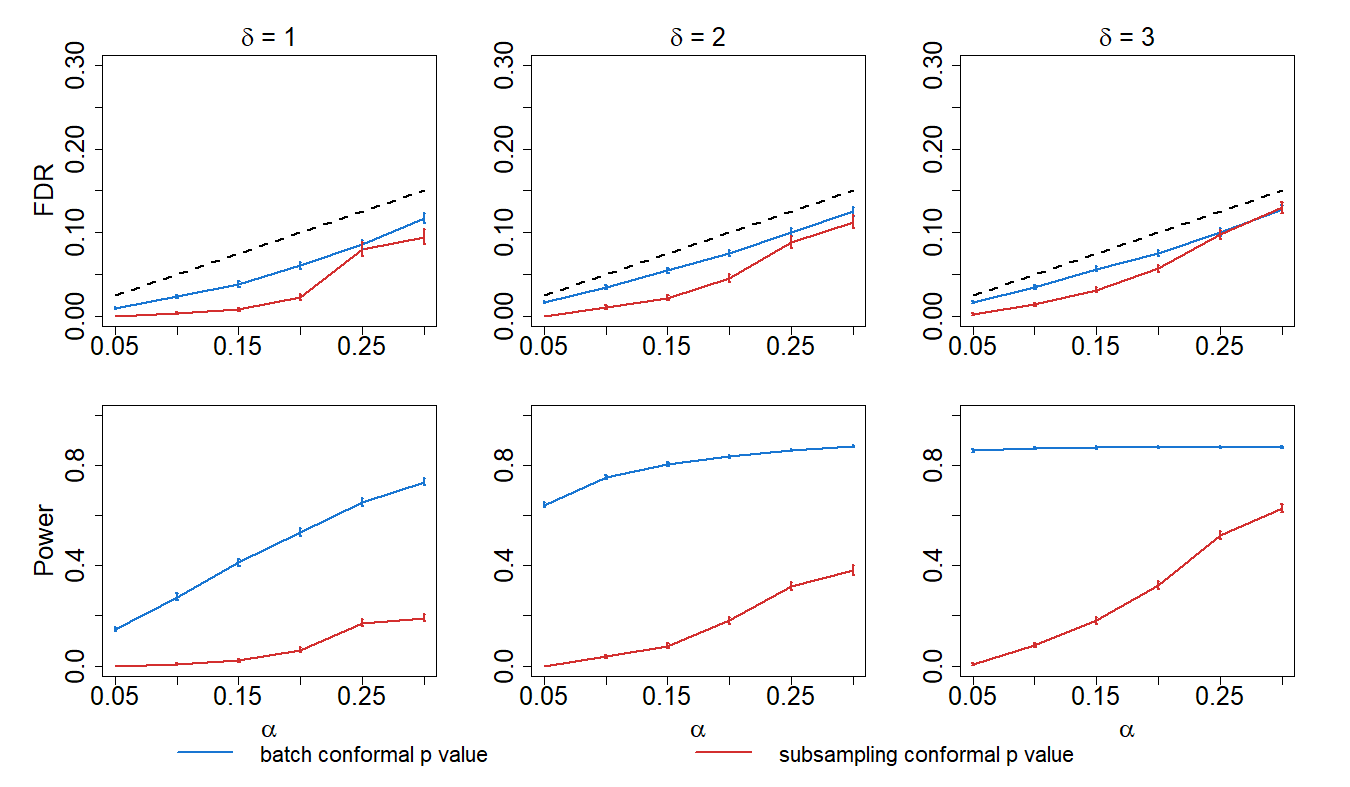}
    \caption{False discovery rate and power of the procedure with batch conformal $p$-values and the procedure with subsampling conformal $p$-values, across different signal strengths. The dotted line corresponds to the theoretical bound from \Cref{thm:prds}.}
    \label{fig:sim_subsampling}
\end{figure}

\subsection{Detecting shifts in a multivariate outcome}\label{sec:sim_mult}

Next, we perform experiments in a setting aimed at detecting shifts in a multivariate outcome. We generate a training set 
\((X_i^\text{train}, Y_i^\text{train})_{1 \leq i \leq 100} \subset \R^{10} \times \R^5\) 
and a calibration dataset 
\((X_i, Y_i)_{1 \leq i \leq 50} \subset \R^{10} \times \R^5\), 
of sizes 100 and 50, respectively, from the following distribution:
\begin{align*}
    &X \sim \textnormal{Unif}([0,1]^{10})\qquad
    (Y_1,Y_2,Y_3) \mid X \sim \mathcal{N}(\beta_1^\top X, \Sigma),\\
    &Y_4 \mid X, Y_1,Y_2,Y_3 \sim \text{Gamma}((Y_1^2+Y_2^2)/2,2),\qquad
    Y_5 \mid X, Y_1,Y_2,Y_3, Y_4 \sim \textnormal{Beta}(|Y_1|, |Y_2|+\tfrac{1}{2}|Y_3|).
\end{align*}
Here, 
each entry of $\beta_1 \in \R^{10 \times 3}$ is generated 
independently from a uniform distribution over $[0,1]$, 
and fixed in advance, while 
$\Sigma$ is a $3 \times 3$ matrix with entries
$\Sigma_{ij} = 2-|i-j|$ for all $i,j \in [3]$.
Next, we generate 50 comparison groups, with group sizes sampled in advance from the distribution 
$5 + \text{Poisson}(20)$. The datapoints in each group are drawn from one of five distributions, 
parametrized by $t \in \{0, 1, 2, 3, 4\}$ assigned to each group:
\begin{align*}
    &X \sim \textnormal{Unif}([0,1]^{10}), \qquad
    (Y_1, Y_2, Y_3) \mid X \sim \mathcal{N}(\beta_1^\top X + t \cdot |\beta_2^\top X|^2, \Sigma), \\
    &Y_4 \mid X, Y_1, Y_2, Y_3 \sim \text{Gamma}((t \cdot |\beta_3^\top X| + Y_1^2 + Y_2^2)/2,2), \qquad
    Y_5 \mid X, Y_1, Y_2, Y_3, Y_4 \sim \textnormal{Beta}(|Y_1|, |Y_2| + \tfrac{1}{2} |Y_3|).
\end{align*}
Here $t = 0$ corresponds to the null distribution. 
The true null proportion is set to 0.5; that is, 25 groups are drawn from the distribution with $t = 0$, 
while $t = 1, 2, 3, 4$ are assigned to ten, five, five, and five groups, respectively. 
We construct estimators $(\hat{\mu}_j)_{j \in [5]}$ of the conditional means of $(Y_j)_{j \in [5]}$ by fitting the training data with random forest regression. Let $\hat{S}$ be the sample covariance matrix of the resulting residual vector of $\hat{r}$ values on the training set.
We run our Algorithm~\ref{alg:batch} with three choices of scores:
\begin{enumerate}
    \item Score A: $s:(x,y) \mapsto \hat{r}(x,y)^\top \hat{S}^{-1} \hat{r}(x,y)$, where 
    \(\hat{r}:(x,y) \mapsto  (y_1 - \hat{\mu}_1(x), y_2 - \hat{\mu}_2(x),y_3 - \hat{\mu}_3(x),y_4 - \hat{\mu}_4(x),y_5 - \hat{\mu}_5(x)).\)
    \item Score B: $s:(x,y) \mapsto \hat{r}(x,y)^\top \hat{S}^{-1} \hat{r}(x,y)$, where 
    \[\hat{r}:(x,y) \mapsto  (y_1 - \hat{\mu}_1(x), y_2 - \hat{\mu}_2(x,y_1),y_3 - \hat{\mu}_3(x,y_1,y_2),y_4 - \hat{\mu}_4(x,y_1,y_2,y_3),y_5 - \hat{\mu}_5(x,y_1,y_2,y_3,y_4)).\]
    \item Score C: $s:(x,y) \mapsto \hat{r}(x,y)^\top \hat{S}^{-1} \hat{r}(x,y)$, where 
    \[\hat{r}:(x,y) \mapsto  (y_5 - \hat{\mu}_1(x), y_4 - \hat{\mu}_2(x,y_5),y_3 - \hat{\mu}_3(x,y_4,y_5),y_2 - \hat{\mu}_4(x,y_3,y_4,y_5),y_1 - \hat{\mu}_5(x,y_2,y_3,y_4,y_5)).\]
\end{enumerate}
Intuitively, score A does not capture the dependence structure between the components of the outcome, 
while scores B and C do.
However, score C uses a mis-specified regression model for the outcomes. 

We repeat the process of generating the reference and comparison datasets and running the procedure 500 times. The results are summarized in Figure~\ref{fig:mult}, illustrating that the proposed procedure controls the false discovery rate at the level provided by Theorem~\ref{thm:prds},
and that the power of the tests is similar for different choice of scores.

\begin{figure}[ht]
    \centering
    \includegraphics[width=0.8\textwidth]{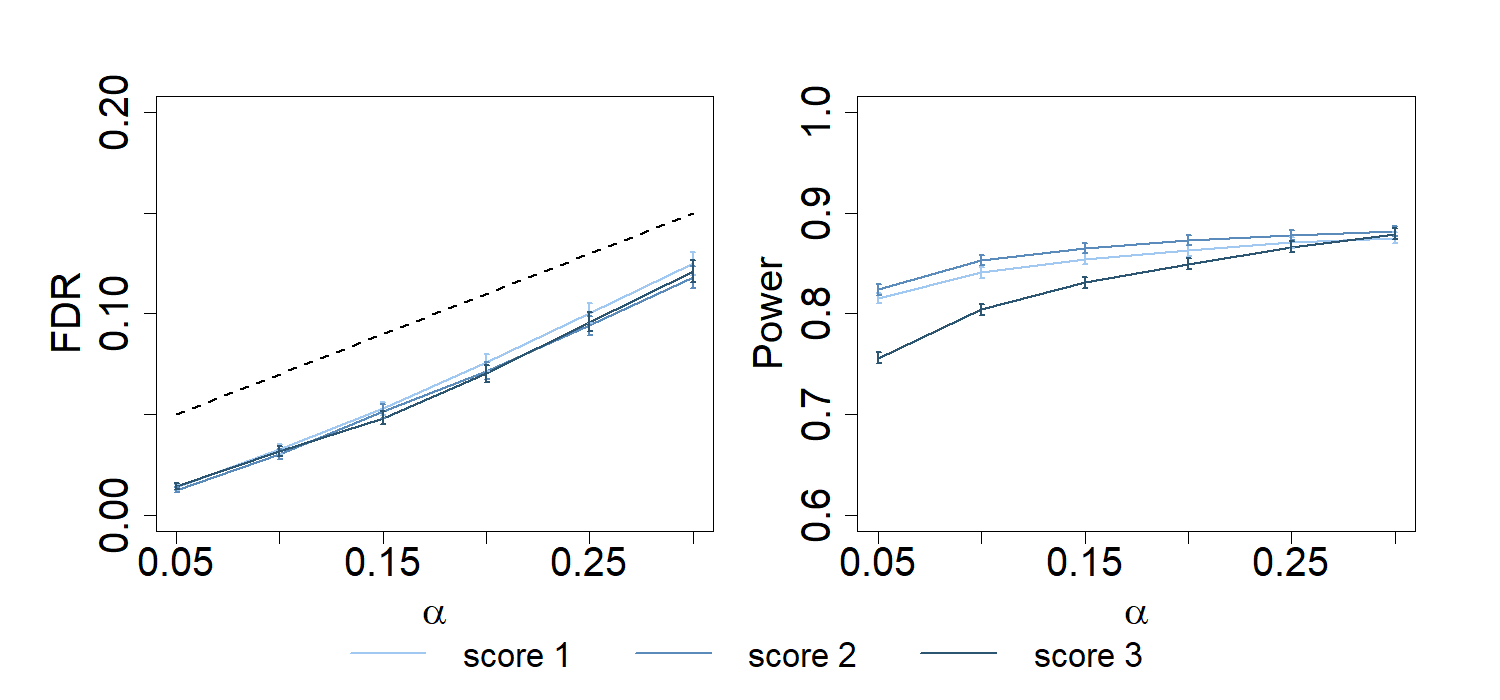}
    \caption{False discovery rate and power of Algorithm~\ref{alg:batch} across different scores. The dotted line corresponds to the theoretical bound $0.5 \alpha$ from \Cref{thm:prds}.}
    \label{fig:mult}
\end{figure}

\subsection{Two-sample test for distributional equality}
\label{ts-de}

Here, we provide simple simulation results for batch conformal p-value-based testing in the two-sample setting, and compare with other distribution-free two-sample tests---the permutation test and the rank-sum test.

We generate reference samples of size $n = 30$ from $\mathcal{N}(0,1)$ and comparison samples of size $m = 30$ from $\mathcal{N}(0,3)$. The two distributions thus differ in scale but not in location---i.e., they have the same mean and median. We then construct three p-values:
\begin{enumerate}
\item Batch conformal $p$-value: computed according to~\eqref{eqn:p_val}, with $\eta = \lceil q \cdot m \rceil$.
\item Permutation test $p$-value: computed according to~\eqref{eqn:p_perm_ran}, with test statistic $T(X_{1:n+m}) = Q_q(X_{1:n}) - Q_q(X_{(n+1):(n+m)})$ and $L = 1000$ permutations.
\item Rank-sum test $p$-value: computed with \texttt{wilcox.test} function.
\end{enumerate}
Here, $Q_q$ denotes the $q$-quantile function. We repeat these steps 500 times and compute the power of the three methods for two choices of $q$: $q = 0.8$ and $q = 0.5$. The former corresponds to comparing the 0.8-quantiles of the two samples, which makes the batch conformal test and the permutation test more likely to detect the difference. The latter represents a ``bad choice”, since the resulting test statistic—the median—is the same for both distributions, making it harder to distinguish between them. The rank-sum test does not involve such a hyperparameter and therefore remains the same in both comparisons.

The results are shown in Figure~\ref{fig:two_sample}. When a tail quantile ($q=0.8$) is used as the test statistic, both the batch conformal p-value-based test and the permutation test show significantly higher power than the rank-sum test---as expected, since the rank-sum test does not capture differences in the tails. Notably, the batch conformal p-value test even achieves higher power than the permutation test. When the median ($q=0.5$) is used, both methods have low power, but still achieve power higher than the rank-sum test. In summary, the test based on the batch conformal p-value has a significant advantage over the rank-sum test when the difference between distributions lies mainly in the tails, and attains power comparable to the permutation test---while being computationally much cheaper.

\begin{figure}[ht]
    \centering
    \includegraphics[width=0.6\textwidth]{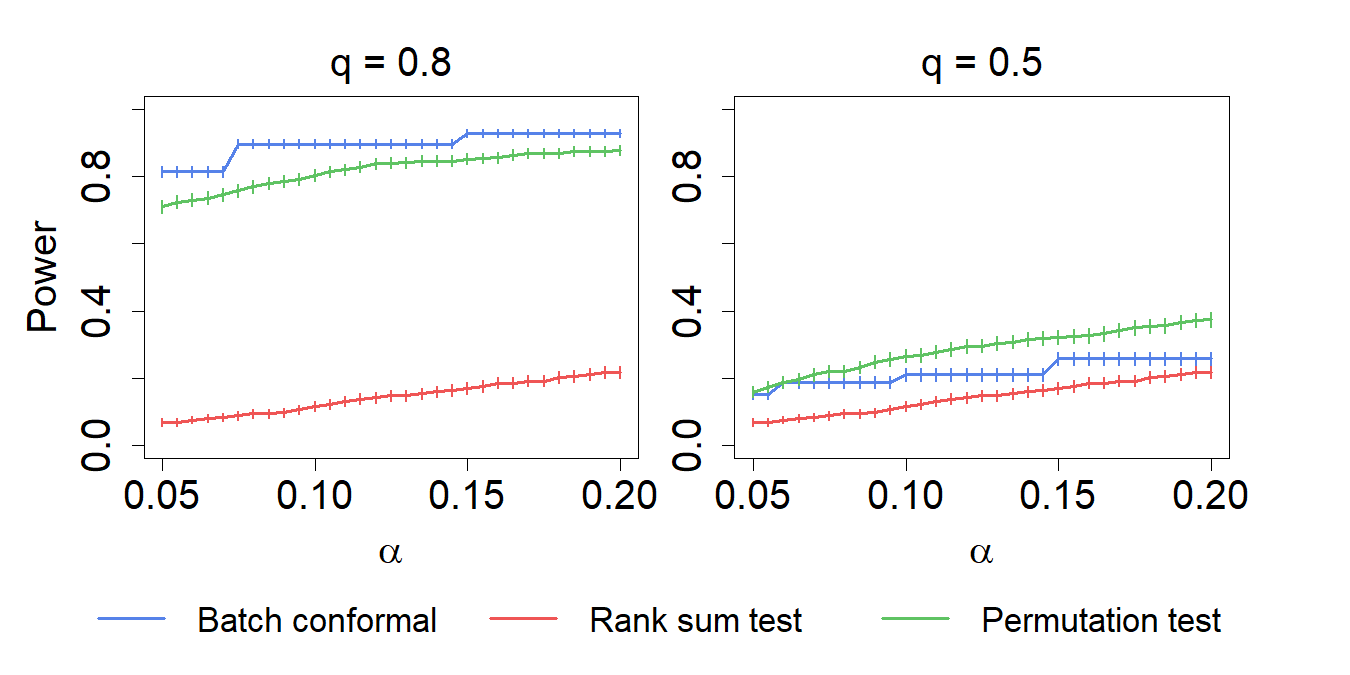}
    \caption{Power of the batch conformal p-value-based test, rank-sum test, and permutation test 
    with two choices of test statistic, along with error bars.}
    \label{fig:two_sample}
\end{figure}

\section{Empirical illustration}

We apply the proposed procedure\footnote{The experiments for the CPS data are provided in Appendix \ref{cps}.} to the HALT-C dataset~\citep{snow2024haltc}, which contains data from patients with chronic hepatitis C who were randomly assigned to 
either
receive peginterferon treatment or not.
The data were collected from ten study sites. We consider the following two tasks:
\begin{enumerate}
    \item Comparison of treatment effects across different patient groups.
    \item Comparison of patient-level baseline characteristics across ten study sites.
\end{enumerate}

In the first task, rather than comparing distributions with observed datapoints, we aim for the more challenging task of comparing within-group treatment effects without direct access to the datapoints, as counterfactuals are unavailable. Details are provided next.

\subsection{Identifying patient groups with larger treatment effects}
We first investigate the effect of the treatment on platelet counts measured after nine months. After removing data points lacking treatment or platelet information, the sample size is 844.

We partition the sample based on age, BMI, and sex. Age is divided into four groups: ``$<40$", ``$40$--$49$", ``$50$--$59$", and ``$\geq60$". Similarly, BMI is categorized into four groups: ``$<25$", ``$25$--$29.9$", ``$30$--$34.9$", and ``$>35$". As a result, the total number of groups formed by these three categorical variables is $4 \cdot 4 \cdot 2 = 32$. Among these groups, we exclude any group with a sample size that is too small---specifically, for which either the control arm or the treated arm has a size less than or equal to one. Out of the remaining 26 groups, we select the group “age 40–49 / BMI 25–29.9 / Male”, which has the largest sample size, as the reference group. The other $K=25$ groups are used as comparison groups. 

Throughout this section, we denote the counterfactual distributions of the comparison groups by $P_{k,0}$ and $P_{k,1}$ for $k \in [K]$, and those of the reference group by $P_{0,0}$ and $P_{0,1}$.
Further, 
we denote the cumulative distribution function of $P_{k,0}$ by $F_{k,0}$ for each $k \in \{0\} \cup [K]$. Let $Y_{k1}, \dots, Y_{k n_k}$ represent the outcomes of the treated individuals in the $k$-th group, and $Y_1, \dots, Y_n$ denote the outcomes in the reference group. 

Assuming for now that the cumulative distribution functions are known, 
construct the scores as $S_i^{(k)} = F_{k,0}(Y_{ki})$ for $i \in [n_k]$ in the comparison groups, and $S_i = F_{0,0}(Y_i)$ for $i \in [n]$ in the reference group. 
The detection procedure with these cumulative distribution function-based scores 
tests
\begin{equation}\label{eqn:null_haltc_cdf}
H_k : F_{k,0}(Y_k) \stackrel{d}{=} F_{0,0}(Y), \text{ where } Y_k \sim P_{k,1} \text{ and } Y \sim P_{0,1}, \quad k \in [K].
\end{equation}
The test based on the batch conformal 
$p$-value rejects the null when 
an appropriate quantile of 
the comparison scores 
$(F_{k,0}(Y_{k,i}))_{i \in [n_k]}$ 
is larger than the corresponding quantile of the reference scores $(F_{0,0}(Y_i))_{i \in [n]}$. 
When the potential outcomes follow continuous distributions, the null hypothesis $H_k$ in~\eqref{eqn:null_haltc_cdf} 
is equivalent to 
\(F_{k,0}\circ ({F_{k,1}}^{-1}) =
F_{0,0} \circ ({F_{0,1}}^{-1})\),
where $\cdot^{-1}$ denotes the inverse function.
Here, 
$F_{k,0}\circ ({F_{k,1}}^{-1})$
is the treated-to-control quantile-quantile map in the $k$-th group, which is also known as the (inverse of the)
\emph{response map} \cite{doksum1974empirical,doksum1976plotting}.
This map has been used to quantify treatment effects in e.g., \cite{athey2023semi}.

In practice, the cumulative distribution functions must be estimated\footnote{
It is also possible to use an external database with a much larger sample size to construct more accurate estimates of the cumulative distribution functions, since only the control group cumulative distribution function is required---e.g., the distribution of standard platelet counts in untreated individuals in our example.} using data from the control arm---e.g., by the empirical cumulative distribution function. 
For each $k \in \{0\}\cup [K]$,
denote 
by $Y_{k1}^{(0)}, \dots, Y_{k n_k^{(0)}}^{(0)}$ 
the outcomes of the control individuals in the $k$-th group, and $Y^{(0)}_1, \dots, Y^{(0)}_{n^{(0)}}$ denote the outcomes in the control arm of the reference group.
For each $k \in \{0\} \cup [K]$,
let the empirical cumulative distribution functions of the outcomes
in the control arm
of group $k$ be $\hat{F}_{k,0}$. 
We can consider 
the following null hypotheses
\emph{conditional on the control arms}---so that the empirical cumulative distribution functions can be treated as fixed:
\begin{equation}\label{eqn:null_empirical_cdf}
   \tilde{H}_k : \hat{F}_{k,0}(Y_k) \stackrel{d}{=} \hat{F}_{0,0}(Y)\, \mid 
   Y_{k1}^{(0)}, \dots, Y_{k n_k^{(0)}}^{(0)},\,
   Y^{(0)}_1, \dots, Y^{(0)}_{n^{(0)}}
   \text{ where } Y_k \sim P_{k,1} \text{ and } Y \sim P_{0,1}, \quad k \in [K].
\end{equation}

Our procedure provides exact 
false discovery rate control for testing the hypotheses~\eqref{eqn:null_empirical_cdf}.
While the
hypotheses in~\eqref{eqn:null_empirical_cdf}
differ from those in~\eqref{eqn:null_haltc_cdf}, 
they nonetheless
capture the same goal---detecting groups with larger treatment effects. 
The formulation in terms of hypotheses~\eqref{eqn:null_empirical_cdf} can be viewed as the finite-population counterpart of the formulation based on~\eqref{eqn:null_haltc_cdf}.

\begin{table}[ht]
\centering
\begin{tabular}{lll|ccc|ccc|ccc}
\specialrule{0.6pt}{0pt}{0pt}
\toprule
& & & \multicolumn{3}{c|}{$Q_1$} & \multicolumn{3}{c|}{$Q_2$} & \multicolumn{3}{c}{$Q_3$} \\
\midrule
Age & BMI & Sex
& 0.05 & 0.1 & 0.2
& 0.05 & 0.1 & 0.2
& 0.05 & 0.1 & 0.2 \\
\midrule
\midrule
$<40$       & $\geq 35$       & M & -- & -- & -- & -- & -- & \checkmark & -- & \checkmark & \checkmark \\
$40$--$49$  & $<25$           & M & -- & -- & -- & -- & -- & -- & -- & \checkmark & \checkmark \\
$40$--$49$  & $25$--$29.9$    & F & -- & -- & -- & -- & -- & -- & \checkmark & \checkmark & \checkmark \\
$40$--$49$  & $30$--$34.9$    & M & -- & \checkmark & \checkmark & -- & -- & -- & -- & -- & -- \\
$40$--$49$  & $30$--$34.9$    & F & -- & -- & \checkmark & -- & -- & -- & -- & -- & -- \\
$40$--$49$  & $\geq 35$       & M & -- & \checkmark & \checkmark & -- & -- & \checkmark & -- & -- & \checkmark \\
$50$--$59$  & $<25$           & F & -- & -- & -- & -- & -- & \checkmark & \checkmark & \checkmark & \checkmark \\
$50$--$59$  & $25$--$29.9$    & M & -- & -- & -- & \checkmark & \checkmark & \checkmark & \checkmark & \checkmark & \checkmark \\
$50$--$59$  & $25$--$29.9$    & F & -- & \checkmark & \checkmark & -- & -- & \checkmark & -- & -- & \checkmark \\
$50$--$59$  & $30$--$34.9$    & M & -- & \checkmark & \checkmark & -- & \checkmark & \checkmark & \checkmark & \checkmark & \checkmark \\
$50$--$59$  & $30$--$34.9$    & F & -- & -- & \checkmark & -- & \checkmark & \checkmark & -- & -- & -- \\
$50$--$59$  & $\geq 35$       & M & -- & -- & -- & -- & -- & \checkmark & -- & \checkmark & \checkmark \\
$\geq 60$   & $30$--$34.9$    & M & \checkmark & \checkmark & \checkmark & \checkmark & \checkmark & \checkmark & \checkmark & \checkmark & \checkmark \\
$\geq 60$   & $30$--$34.9$    & F & -- & \checkmark & \checkmark & -- & -- & -- & -- & -- & -- \\
$\geq 60$   & $\geq 35$       & F & -- & -- & -- & \checkmark & \checkmark & \checkmark & \checkmark & \checkmark & \checkmark \\
\specialrule{0.6pt}{0pt}{0pt}
\bottomrule
\end{tabular}
\caption{Results for the HALT-C dataset: selected groups at levels $\alpha = 0.05, 0.1$, and $0.2$, out of 25 groups based on age, BMI, and sex.  The reference group corresponds to age 40–49/BMI 25–29.9/Male.}
\label{tab:haltc_cdf}
\end{table}

Table~\ref{tab:haltc_cdf} shows the results from the method using cumulative distribution function-based scores, where, as in the previous experiment, the two quartiles and the median are used as test statistics. The number of detected groups is larger for higher quantiles, and especially for $Q_3$, suggesting that there are larger differences in treatment effects between the comparison and reference groups at larger quantiles of the treated outcomes.

\subsection{Identifying study sites with different population characteristics}

Next, we compare the population composition across the ten study sites in terms of patients' baseline characteristics. We conduct the test based on age, BMI, sex, and baseline platelet counts of the individuals---distinct from the ``platelet counts after treatment''  used in the previous experiment. The study site with the largest sample size (284) is chosen as the reference group, and the remaining nine sites are set as the comparison groups. In other words, denoting the joint distribution of age, BMI, sex, and platelet counts for the $k$-th comparison site as $P^{(k)}$ and for the reference site as $P$, we test
\[H_k : P^{(k)} = P, \qquad k=1,2,\cdots,9.\]

To construct the score function, we split the reference data into two sets, each of size 142, and use one split to compute the estimates $\hat{\mu} = (\hat{\mu}_\text{age}, \hat{\mu}_\text{BMI}, \hat{\mu}_\text{sex}, \hat{\mu}_\text{pl})$, which are simply the sample means of age, BMI, sex (encoded as a binary variable), and platelet counts, respectively. We also compute the sample covariance matrix $\hat{S}$ on the same split, and define the score function as
\[s(x) = s((x_\text{age}, x_\text{BMI}, x_\text{sex}, x_\text{pl})) = (x - \hat{\mu})^\top \hat{S}^{-1} (x - \hat{\mu}).\]
The second split is used to compute the $p$-values along with the comparison datasets.

Figure~\ref{fig:haltc_sites} presents the number of rejections from Algorithm~\ref{alg:batch} at levels $\alpha = 0.01, 0.02, \dots, 1$, with three different choices of the test statistic: the lower and upper quantiles, and the median. Specifically, for each group, we construct the $p$-value using $\eta_k = \lceil 0.25 \cdot n_k \rceil$, $\lceil 0.5 \cdot n_k \rceil$, or $\lfloor 0.75 \cdot n_k \rfloor$ in the formula~\eqref{eqn:p_val_k}.
None of the comparison sites is concluded to have a different distribution from the reference site unless the target false discovery rate control level $\alpha$ is set to a large value, which is not typically used in multiple testing procedures.
Thus, no significant differences 
have been detected
in the population composition across the ten sites, 
providing one component of justification for using a site-pooled sample in various analyses.

\begin{figure}[ht]
    \centering
    \includegraphics[width=0.85\textwidth]{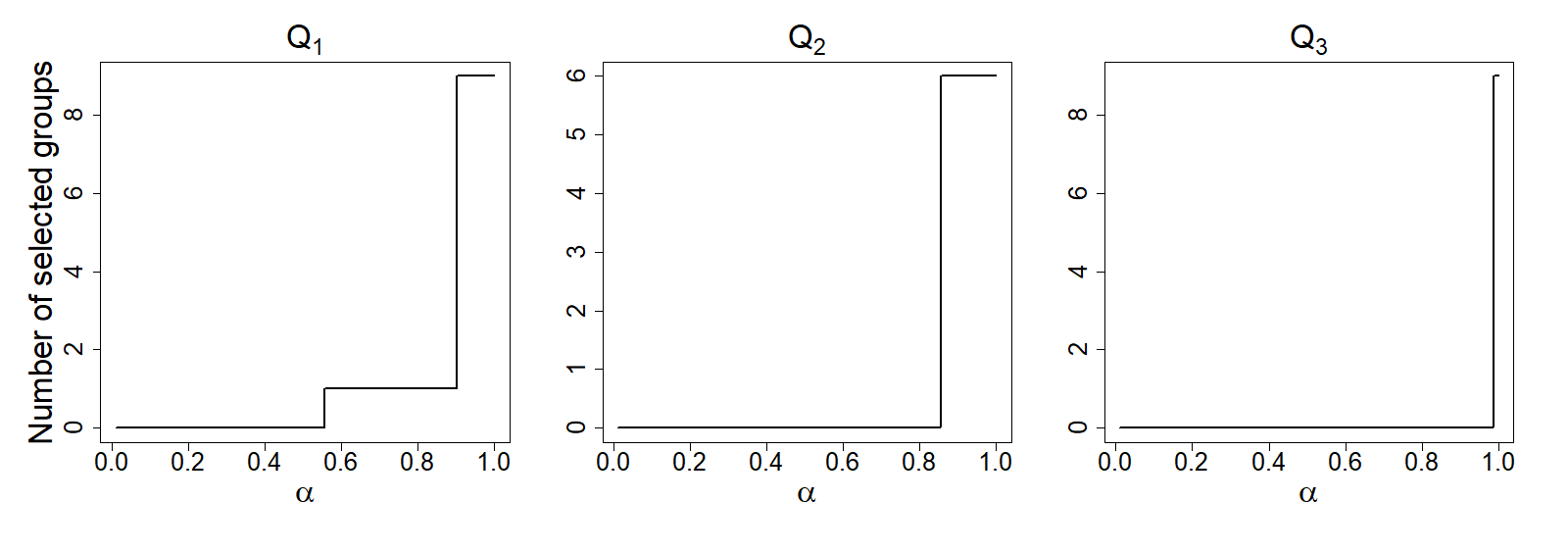}
    \caption{Results for the HALT-C dataset: the number of selected study sites from Algorithm~\ref{alg:batch}, across different values of the level $\alpha$.}
    \label{fig:haltc_sites}
\end{figure}

\section{Discussion}

In this work, we explore the problem of identifying distributional shifts across multiple comparison datasets and provide a methodology that offers distribution-free false discovery rate control while being computationally lightweight. 
The proposed methodology allows the repeated use of the reference dataset, by making use of positive regression dependence of the batch conformal $p$-values. Empirical results illustrate that this approach indeed controls the false discovery rate with reasonable power.

Several open questions remain regarding the application of batch conformal $p$-values.
For instance, 
it can be of interest to test for pairs of groups, beyond comparing every dataset with one reference dataset \citep[e.g.,][etc]{dunnett1955multiple}. 
Is there a way to re-use the data in each group to construct $p$-values, and can we construct a multiple testing procedure that outputs a consistent decision---in the sense that the individual decisions do not contradict each other? 
This question is left for future exploration.

\section*{Acknowledgement}
This work was supported in part by 
NIH R01-AG065276, R01-GM139926, NSF 2210662, P01-AG041710, R01-CA222147, as well as the ARO, ONR, and the Sloan Foundation.

\bibliographystyle{plainnat}
\bibliography{bib}

@article{ramdas2017wasserstein,
  title={On wasserstein two-sample testing and related families of nonparametric tests},
  author={Ramdas, Aaditya and Garc{\'\i}a Trillos, Nicol{\'a}s and Cuturi, Marco},
  journal={Entropy},
  volume={19},
  number={2},
  pages={47},
  year={2017},
  publisher={MDPI}
}

@article{stranger2007population,
  title={Population genomics of human gene expression},
  author={Stranger, Barbara E and Nica, Alexandra C and Forrest, Matthew S and Dimas, Antigone and Bird, Christine P and Beazley, Claude and Ingle, Catherine E and Dunning, Mark and Flicek, Paul and Koller, Daphne and others},
  journal={Nature genetics},
  volume={39},
  number={10},
  pages={1217--1224},
  year={2007},
  publisher={Nature Publishing Group US New York}
}

@article{salojarvi2017genome,
  title={Genome sequencing and population genomic analyses provide insights into the adaptive landscape of silver birch},
  author={Saloj{\"a}rvi, Jarkko and Smolander, Olli-Pekka and Nieminen, Kaisa and Rajaraman, Sitaram and Safronov, Omid and Safdari, Pezhman and Lamminm{\"a}ki, Airi and Immanen, Juha and Lan, Tianying and Tanskanen, Jaakko and others},
  journal={Nature genetics},
  volume={49},
  number={6},
  pages={904--912},
  year={2017},
  publisher={Nature Publishing Group}
}

@article{maura2022efficient,
    author = {John, Maura and Ankenbrand, Markus J and Artmann, Carolin and Freudenthal, Jan A and Korte, Arthur and Grimm, Dominik G},
    title = {Efficient permutation-based genome-wide association studies for normal and skewed phenotypic distributions},
    journal = {Bioinformatics},
    volume = {38},
    number = {Supplement\_2},
    pages = {ii5-ii12},
    year = {2022},
    month = {09}
}

@article{mary2022semi,
  title={Semi-supervised multiple testing},
  author={Mary, David and Roquain, Etienne},
  journal={Electronic Journal of Statistics},
  volume={16},
  number={2},
  pages={4926--4981},
  year={2022},
  publisher={The Institute of Mathematical Statistics and the Bernoulli Society}
}

@book{wilks1962mathematical,
  title={Mathematical statistics},
  author={Wilks, Samuel S},
  year={1962},
  publisher={Wiley}
}

@article{yang2017smoothed,
  title={Smoothed jackknife empirical likelihood for the difference of two quantiles},
  author={Yang, Hanfang and Zhao, Yichuan},
  journal={Annals of the Institute of Statistical Mathematics},
  volume={69},
  pages={1059--1073},
  year={2017},
  publisher={Springer}
}

@article{lou2025high,
  title={High-dimensional Simultaneous Inference of Quantiles},
  author={Lou, Zhipeng and Wu, Wei Biao},
  journal={Sankhya A},
  pages={1--27},
  year={2025},
  publisher={Springer}
}

@article{sgouropoulos2015matching,
  title={Matching a distribution by matching quantiles estimation},
  author={Sgouropoulos, Nikolaos and Yao, Qiwei and Yastremiz, Claudia},
  journal={Journal of the American Statistical Association},
  volume={110},
  number={510},
  pages={742--759},
  year={2015},
  publisher={Taylor \& Francis}
}

@article{kosorok1999two,
  title={Two-sample quantile tests under general conditions},
  author={Kosorok, Michael R},
  journal={Biometrika},
  volume={86},
  number={4},
  pages={909--921},
  year={1999},
  publisher={Oxford University Press}
}

@article{benjamini2010discovering,
  title={Discovering the false discovery rate},
  author={Benjamini, Yoav},
  journal={Journal of the Royal Statistical Society Series B: Statistical Methodology},
  volume={72},
  number={4},
  pages={405--416},
  year={2010},
  publisher={Oxford University Press}
}

@misc{census_cps,
  author       = {{U.S. Census Bureau}},
  title        = {Current Population Survey (CPS) Datasets},
  year         = {2024},
  url          = {https://www.census.gov/programs-surveys/cps/data/datasets.html},
  note         = {Accessed: 2025-05-19}
}

@article{doksum1976plotting,
  title={Plotting with confidence: Graphical comparisons of two populations},
  author={Doksum, Kjell A and Sievers, Gerald L},
  journal={Biometrika},
  volume={63},
  number={3},
  pages={421--434},
  year={1976},
  publisher={Oxford University Press}
}

@article{doksum1974empirical,
  title={Empirical probability plots and statistical inference for nonlinear models in the two-sample case},
  author={Doksum, Kjell},
  journal={The annals of statistics},
  pages={267--277},
  year={1974},
  publisher={JSTOR}
}

@article{athey2023semi,
  title={Semi-parametric estimation of treatment effects in randomised experiments},
  author={Athey, Susan and Bickel, Peter J and Chen, Aiyou and Imbens, Guido W and Pollmann, Michael},
  journal={Journal of the Royal Statistical Society Series B: Statistical Methodology},
  volume={85},
  number={5},
  pages={1615--1638},
  year={2023},
  publisher={Oxford University Press US}
}

@inproceedings{romano2019conformalized,
  author = {Romano, Yaniv and Patterson, Evan and Candes, Emmanuel},
  booktitle = {{A}dvances in {N}eural {I}nformation {P}rocessing {S}ystems},
  title = {Conformalized quantile regression},
  year = {2019}}

@article{kuchibhotla2020exchangeability,
  title={Exchangeability, conformal prediction, and rank tests},
  author={Kuchibhotla, Arun Kumar},
  journal={arXiv preprint arXiv:2005.06095},
  year={2020}
}

@book{lehmann2006nonparametrics,
  title={Nonparametrics: statistical methods based on ranks},
  author={Lehmann, Erich Leo and D'Abrera, Howard JM},
  year={2006},
  publisher={Springer New York}
}

@article{vovk2021values,
  title={E-values: Calibration, combination and applications},
  author={Vovk, Vladimir and Wang, Ruodu},
  journal={The Annals of Statistics},
  volume={49},
  number={3},
  pages={1736--1754},
  year={2021},
  publisher={Institute of Mathematical Statistics}
}

@book{lehmann2022testing,
  title={Testing Statistical Hypotheses},
  author={Lehmann, Erich L and Romano, Joseph P},
  year={2022},
  publisher={Springer Science \& Business Media}
}

@article{shaffer1995multiple,
  title={Multiple hypothesis testing},
  author={Shaffer, Juliet Popper},
  journal={Annual review of psychology},
  volume={46},
  number={1},
  pages={561--584},
  year={1995}
}

@book{pesarin2001multivariate,
  title={Multivariate permutation tests: with applications in biostatistics},
  author={Pesarin, Fortunato},
  year={2001},
  publisher={Wiley Chichester}
}

@incollection{lehmann2012parametric,
  title={Parametric versus nonparametrics: two alternative methodologies},
  author={Lehmann, Erich L},
  booktitle={Selected works of EL Lehmann},
  pages={437--445},
  year={2012},
  publisher={Springer}
}

@article{kim2020robust,
  title={Robust multivariate nonparametric tests via projection averaging},
  author={Kim, Ilmun and Balakrishnan, Sivaraman and Wasserman, Larry},
  journal={Annals of Statistics},
  volume={48},
  number={6},
  pages={3417--3441},
  year={2020},
  publisher={Institute of Mathematical Statistics}
}

@article{kim2020minimax,
  title={Minimax optimality of permutation tests},
  author={Kim, Ilmun and Balakrishnan, Sivaraman and Wasserman, Larry},
  journal={arXiv preprint arXiv:2003.13208},
  year={2020}
}

@article{pesarin2010finite,
  title={Finite-sample consistency of combination-based permutation tests with application to repeated measures designs},
  author={Pesarin, Fortunato and Salmaso, Luigi},
  journal={Journal of Nonparametric Statistics},
  volume={22},
  number={5},
  pages={669--684},
  year={2010},
  publisher={Taylor \& Francis}
}

@article{pesarin2013weak,
  title={On the weak consistency of permutation tests},
  author={Pesarin, F and Salmaso, L},
  journal={Communications in Statistics-Simulation and Computation},
  volume={42},
  number={6},
  pages={1368--1379},
  year={2013},
  publisher={Taylor \& Francis}
}

@article{dwass1957modified,
  title={Modified randomization tests for nonparametric hypotheses},
  author={Dwass, Meyer},
  journal={The Annals of Mathematical Statistics},
  pages={181--187},
  year={1957},
  publisher={JSTOR}
}

@article{ernst2004permutation,
  title={Permutation methods: a basis for exact inference},
  author={Ernst, Michael D},
  journal={Statistical Science},
  volume={19},
  number={4},
  pages={676--685},
  year={2004},
  publisher={Institute of Mathematical Statistics}
}

@article{hemerik2018exact,
  title={Exact testing with random permutations},
  author={Hemerik, Jesse and Goeman, Jelle},
  journal={Test},
  volume={27},
  number={4},
  pages={811--825},
  year={2018},
  publisher={Springer}
}

@article{pesarin2015some,
  title={Some elementary theory of permutation tests},
  author={Pesarin, Fortunato},
  journal={Communications in Statistics-Theory and Methods},
  volume={44},
  number={22},
  pages={4880--4892},
  year={2015},
  publisher={Taylor \& Francis}
}

@article{kennedy1995randomization,
  title={Randomization tests in econometrics},
  author={Kennedy, Fetter E},
  journal={Journal of Business \& Economic Statistics},
  volume={13},
  number={1},
  pages={85--94},
  year={1995},
  publisher={Taylor \& Francis}
}

@article{anderson2001permutation,
  title={Permutation tests for univariate or multivariate analysis of variance and regression},
  author={Anderson, Marti J},
  journal={Canadian journal of fisheries and aquatic sciences},
  volume={58},
  number={3},
  pages={626--639},
  year={2001},
  publisher={NRC Research Press Ottawa, Canada}
}

@article{pitman1937significance,
  title={Significance tests which may be applied to samples from any populations},
  author={Pitman, Edwin JG},
  journal={Supplement to the Journal of the Royal Statistical Society},
  volume={4},
  number={1},
  pages={119--130},
  year={1937},
  publisher={JSTOR}
}

@book{good2006permutation,
  title={Permutation, parametric, and bootstrap tests of hypotheses},
  author={Good, Phillip I},
  year={2006},
  publisher={Springer Science \& Business Media}
}

@book{fisher1935design,
  title={The design of experiments},
  author={Fisher, Ronald Aylmer},
  year={1935},
  publisher={Oliver and Boyd}
}

@article{eden1933validity,
  title={On the validity of Fisher's z test when applied to an actual example of non-normal data},
  author={Eden, T and Yates, F},
  journal={The Journal of Agricultural Science},
  volume={23},
  number={1},
  pages={6--17},
  year={1933},
  publisher={Cambridge University Press}
}

@article{pesarin2012review,
  title={A review and some new results on permutation testing for multivariate problems},
  author={Pesarin, Fortunato and Salmaso, Luigi},
  journal={Statistics and Computing},
  volume={22},
  number={2},
  pages={639--646},
  year={2012},
  publisher={Springer}
}

@article{pesarin2010permutation,
  title={The permutation testing approach: a review},
  author={Pesarin, Fortunato and Salmaso, Luigi},
  journal={Statistica},
  volume={70},
  number={4},
  pages={481--509},
  year={2010}
}

@article{strassen1965existence,
  title={The existence of probability measures with given marginals},
  author={Strassen, Volker},
  journal={The Annals of Mathematical Statistics},
  volume={36},
  number={2},
  pages={423--439},
  year={1965},
  publisher={Institute of Mathematical Statistics}
}

@article{lee2024batch,
  title={Batch Predictive Inference},
  author={Lee, Yonghoon and Tchetgen, Eric Tchetgen and Dobriban, Edgar},
  journal={arXiv preprint arXiv:2409.13990},
  year={2024}
}

@article{benjamini2001control,
  title={The control of the false discovery rate in multiple testing under dependency},
  author={Benjamini, Yoav and Yekutieli, Daniel},
  journal={Annals of statistics},
  pages={1165--1188},
  year={2001},
  publisher={JSTOR}
}

@article{bates2023testing,
  title={Testing for outliers with conformal p-values},
  author={Bates, Stephen and Cand{\`e}s, Emmanuel and Lei, Lihua and Romano, Yaniv and Sesia, Matteo},
  journal={The Annals of Statistics},
  volume={51},
  number={1},
  pages={149--178},
  year={2023},
  publisher={Institute of Mathematical Statistics}
}

@article{biswas2014distribution,
    author = {Biswas, Munmun and Mukhopadhyay, Minerva and Ghosh, Anil K.},
    title = {A distribution-free two-sample run test applicable to high-dimensional data},
    journal = {Biometrika},
    volume = {101},
    number = {4},
    pages = {913-926},
    year = {2014},
    month = {10}
}

@article{rosenbaum2005exact,
  title={An exact distribution-free test comparing two multivariate distributions based on adjacency},
  author={Rosenbaum, Paul R},
  journal={Journal of the Royal Statistical Society Series B: Statistical Methodology},
  volume={67},
  number={4},
  pages={515--530},
  year={2005},
  publisher={Oxford University Press}
}

@article{mann1947test,
  title={On a test of whether one of two random variables is stochastically larger than the other},
  author={Mann, Henry B and Whitney, Donald R},
  journal={The annals of mathematical statistics},
  pages={50--60},
  year={1947},
  publisher={JSTOR}
}

@article{benjamini1995controlling,
  title={Controlling the false discovery rate: a practical and powerful approach to multiple testing},
  author={Benjamini, Yoav and Hochberg, Yosef},
  journal={Journal of the Royal statistical society: series B (Methodological)},
  volume={57},
  number={1},
  pages={289--300},
  year={1995},
  publisher={Wiley Online Library}
}

@article{dunnett1955multiple,
  title={A multiple comparison procedure for comparing several treatments with a control},
  author={Dunnett, Charles W},
  journal={Journal of the American Statistical Association},
  volume={50},
  number={272},
  pages={1096--1121},
  year={1955},
  publisher={Taylor \& Francis}
}

@article{hothorn2020comparisons,
  title={Comparisons of multiple treatment groups with a negative control or placebo group: Dunnett test vs. closed test procedure},
  author={Hothorn, Ludwig A},
  journal={arXiv preprint arXiv:2012.04277},
  year={2020}
}

@article{marandon2024adaptive,
  title={Adaptive novelty detection with false discovery rate guarantee},
  author={Marandon, Ariane and Lei, Lihua and Mary, David and Roquain, Etienne},
  journal={The Annals of Statistics},
  volume={52},
  number={1},
  pages={157--183},
  year={2024},
  publisher={Institute of Mathematical Statistics}
}

@article{magnani2024collective,
  title={Collective Outlier Detection and Enumeration with Conformalized Closed Testing},
  author={Magnani, Chiara G and Sesia, Matteo and Solari, Aldo},
  journal={Proceedings of Machine Learning Research},
  volume={230},
  pages={1--1},
  year={2024}
}

@article{lee2025full,
  title={Full-conformal novelty detection: A powerful and non-random approach},
  author={Lee, Junu and Popov, Ilia and Ren, Zhimei},
  journal={arXiv preprint arXiv:2501.02703},
  year={2025}
}

@book{vovk2005algorithmic,
  title={Algorithmic learning in a random world},
  author={Vovk, Vladimir and Gammerman, Alexander and Shafer, Glenn},
  year={2005},
  publisher={Springer}
}

@article{angelopoulos2021gentle,
  title={A gentle introduction to conformal prediction and distribution-free uncertainty quantification},
  author={Angelopoulos, Anastasios N and Bates, Stephen},
  journal={arXiv preprint arXiv:2107.07511},
  year={2021}
}

@article{an1933sulla,
  title={Sulla determinazione empirica di una legge di distribuzione},
  author={Kolmogorov, Andrei N.},
  journal={Giorn Dell'inst Ital Degli Att},
  volume={4},
  pages={89--91},
  year={1933}
}

@article{smirnov1948table,
  title={Table for estimating the goodness of fit of empirical distributions},
  author={Smirnov, Nickolay},
  journal={The annals of mathematical statistics},
  volume={19},
  number={2},
  pages={279--281},
  year={1948},
  publisher={Institute of Mathematical Statistics}
}

@article{szekely2004testing,
  title={Testing for equal distributions in high dimension},
  author={Sz{\'e}kely, G{\'a}bor J and Rizzo, Maria L},
  journal={InterStat},
  volume={5},
  number={16.10},
  pages={1249--1272},
  year={2004},
  publisher={Citeseer}
}

@article{gretton2012kernel,
  title={A kernel two-sample test},
  author={Gretton, Arthur and Borgwardt, Karsten M and Rasch, Malte J and Sch{\"o}lkopf, Bernhard and Smola, Alexander},
  journal={The Journal of Machine Learning Research},
  volume={13},
  number={1},
  pages={723--773},
  year={2012},
  publisher={JMLR. org}
}

@article{student1908probable,
  title={The probable error of a mean},
  author={Student},
  journal={Biometrika},
  pages={1--25},
  year={1908},
  volume={6},
  number={1},
  publisher={JSTOR}
}

@article{sarkar2002some,
  title={Some results on false discovery rate in stepwise multiple testing procedures},
  author={Sarkar, Sanat K},
  journal={The Annals of Statistics},
  volume={30},
  number={1},
  pages={239--257},
  year={2002},
  publisher={Institute of Mathematical Statistics}
}

@article{blanchard2008two,
author = {Gilles Blanchard and Etienne Roquain},
title = {{Two simple sufficient conditions for FDR control}},
volume = {2},
journal = {Electronic Journal of Statistics},
number = {none},
publisher = {Institute of Mathematical Statistics and Bernoulli Society},
pages = {963 -- 992},
year = {2008},
}

@article{benjamini2006adaptive,
  title={Adaptive linear step-up procedures that control the false discovery rate},
  author={Benjamini, Yoav and Krieger, Abba M and Yekutieli, Daniel},
  journal={Biometrika},
  volume={93},
  number={3},
  pages={491--507},
  year={2006},
  publisher={Oxford University Press}
}

@misc{snow2024haltc,
  author       = {Kristin Snow},
  title        = {The Hepatitis C Antiviral Long-Term Treatment Against Cirrhosis Trial (HALT-C) (Version 7)},
  year         = {2024},
  howpublished = {NIDDK Central Repository},
  doi          = {10.58020/d1m7-ye92},
  url          = {https://doi.org/10.58020/d1m7-ye92}
}

@incollection{sarkar2008simes,
  title={On the Simes inequality and its generalization},
  author={Sarkar, Sanat K},
  booktitle={Beyond parametrics in interdisciplinary research: Festschrift in honor of Professor Pranab K. Sen},
  volume={1},
  pages={231--243},
  year={2008},
  publisher={Institute of Mathematical Statistics}
}

\appendix
\section{Details for two-sample testing}
\label{2s-det}

\subsection{Permutation test} 

Permutation tests are a tool for assessing whether two distributions are equal or, more generally, whether two datasets are exchangeable. Given a test statistic $T : \X^{n+m} \rightarrow \R$, the $p$-value of the permutation test is 
\begin{equation}\label{eqn:p_perm}
    p = \frac{\sum_{\sigma \in \mathcal{S}_{n+m}} \One{T(X_{\sigma(1)}, X_{\sigma(2)}, \ldots, X_{\sigma(n+m)}) \geq T(X_1,X_2,\ldots,X_{n+m})}}{(n+m)!}.
\end{equation}
The test statistic $T$ can be any function---popular choices include the mean difference $T(x_1,\ldots,x_{n+m}) = |\frac{1}{m}\sum_{j=1}^m x_{n+j} - \frac{1}{n}\sum_{i=1}^n x_i|$
and the quantile difference $T(x_1,\ldots,x_{n+m}) = |Q_\tau(\{x_{n+1}, \ldots, x_{n+m}\})$ $ - Q_\tau(\{x_1, $ $\ldots, x_n\})$.

In practice, directly computing the $p$-value based on all $(n+m)!$ permutations is often computationally infeasible. Therefore, the following randomized version of the $p$-value is commonly used \citep{dwass1957modified,lehmann2022testing}:
\begin{equation}\label{eqn:p_perm_ran}
    p = \frac{1 + \sum_{l=1}^L \One{T(X_{\sigma_l(1)}, X_{\sigma_l(2)}, \ldots, X_{\sigma_l(n+m)}) \geq T(X_1,X_2,\ldots,X_{n+m})}}{L+1},
\end{equation}
where $\{\sigma_1,\ldots, \sigma_L\}$ is a random sample of permutations drawn uniformly from $\mathcal{S}_{n+m}$. Both the statistics~\eqref{eqn:p_perm} and~\eqref{eqn:p_perm_ran} are valid in a distribution-free sense---i.e., under any distribution $P$, the inequality $\PP{p \leq \alpha} \leq \alpha$ holds for any $\alpha \in (0, 1)$ under the null hypothesis $H_0$ in~\eqref{eqn:null}, see, for example, \citet{lehmann2012parametric}.

\subsection{Rank-sum test and other methods}

If $\mathcal{X}\subset \R$,
the rank-sum test~\citep{mann1947test} uses a rank-based test statistic, 
\(U = nm + n(n+1)/{2} - R\),
where $R$ is the sum of ranks of the test observations $X_{n+1}, \dots, X_{n+m}$ within the set of all observations ${X_1, \dots, X_{n+m}}$. If $\X$ is not a fully ordered set, each observation can first be mapped to a real number before applying the procedure. For small sample sizes $n$ and $m$, the cumulative distribution function $F_U$ of $U$ can be computed directly to obtain a $p$-value, $F_U(U)$, with exact finite-sample validity. However, for large sample sizes, a normal approximation with an asymptotic guarantee is often used \citep{lehmann2006nonparametrics}.
Extensions include \citet{rosenbaum2005exact} and \citet{biswas2014distribution}, among many others.

\section{Testing based on multiple quantiles}\label{sec:multi_quantile}

The testing procedure based on the batch conformal $p$-value in~\eqref{eqn:p_val} rejects the null hypothesis when a specified quantile of the test scores is unusually large compared to the reference scores.
Consequently, if the distribution of the median does not change strongly under the alternative---e.g., it mainly changes in the tails---then the test's power might be low. 
If it is unknown where the shift is likely to occur, 
one might consider using multiple quantiles, such as the three quartiles, to capture the shift in an unknown region.
Can we construct a valid $p$-value based on multiple quantiles of the test scores? We provide a positive answer next.

For clarity and intuition, we present a methodology for testing with two quantiles---generalization beyond two is straightforward. 
Given $0<\eta_1 < \eta_2<1$,
suppose we aim to reject the null if either $S_{(\eta_1)}^\text{test}$ or $S_{(\eta_2)}^\text{test}$ is large compared to the reference scores. 
Let $\tilde{\eta}_1 = \round(\eta_1 \cdot \frac{n}{m})$ and $\tilde{\eta}_2 = \round(\eta_2 \cdot \frac{n}{m})$ be the corresponding scaled ranks in the reference dataset, where $\round(\cdot)$ rounds real numbers to the nearest integer, with $\round(j+1/2)=j$ for any integer $j$. 
We will design a test that rejects the null if
either 
(1) $S_{(\eta_1)}^\text{test}$ is large compared to $S_{(\tilde{\eta}_1)}$ or 
(2) $S_{(\eta_2)}^\text{test}$ is large compared to $S_{(\tilde{\eta}_2)}$.

We can construct the $p$-value as
\begin{equation}\label{eqn:p_two_quantiles}
\begin{split}
    &p = \sum_{t = -\tilde{\eta}_1+1}^{n-\tilde{\eta}_2+1} w_t \cdot \One{S_{(\eta_1)}^\text{test} \leq S_{(\tilde{\eta}_1+t)}, S_{(\eta_2)}^\text{test} \leq S_{(\tilde{\eta}_2+t)}}, \textnormal{ where }\\
    &w_t = \sum_{\substack{t_1, t_2 \,:\, \max\{t_1,t_2\} = t \\ 1 \leq \eta_1+\tilde{\eta}_1+t_1 + 1 < \eta_2+\tilde{\eta}_2+t_2 + 1 \leq n+m}} \frac{\binom{\eta_1+\tilde{\eta}_1+t_1-2}{\eta_1-1}\binom{\eta_2+\tilde{\eta}_2+t_2-\eta_1+\tilde{\eta}_1+t_1-1}{\eta_2-\eta_1-1}\binom{n+m-\eta_2-\tilde{\eta}_2-t_2+1}{m-\eta_2}}{\binom{n+m}{m}}.
\end{split}
\end{equation}
Here, we define $S_{(n+1)} = +\infty$. Observe that the $p$ above tends to become small when either $S_{(\eta_1)}^\text{test}$ or $S_{(\eta_2)}^\text{test}$ is large.
The following result shows the validity of this method.

\begin{theorem}\label{thm:two_quantiles}
    The statistic $p$ in~\eqref{eqn:p_two_quantiles} is a valid $p$-value for the null hypothesis $H_0$ in~\eqref{eqn:null}.
\end{theorem}

The proof is presented in \Cref{pfthm:two_quantiles}.

\section{Additional experimental results}

\subsection{Comparison with permutation-test-based method}

Here, we repeat the simulation in Section~\ref{sec:sim_mult} with additional baseline methods. We fix the score to score A and compare the following three methods:
\begin{enumerate}
\item the Benjamini–Hochberg procedure applied to the batch conformal $p$-values (proposed procedure),
\item the Benjamini–Hochberg procedure applied to the permutation test $p$-values (no theoretical guarantee), and
\item the Benjamini–Yekutieli procedure applied to the permutation test $p$-values.
\end{enumerate}
The results are shown in Figure~\ref{fig:mult_perm}. As discussed in Section~\ref{sec:BY}, the Benjamini–Yekutieli procedure provides conservative detection with low power. The Benjamini–Hochberg procedure applied to the permutation test $p$-values shows results comparable to those of the proposed procedure---although it has no theoretical guarantee and is computationally much heavier. In this setting, the proposed procedure can be viewed as a `correction' of this method that enables provable false discovery rate control while being computationally cheaper.

\begin{figure}[ht]
    \centering
    \includegraphics[width=0.8\textwidth]{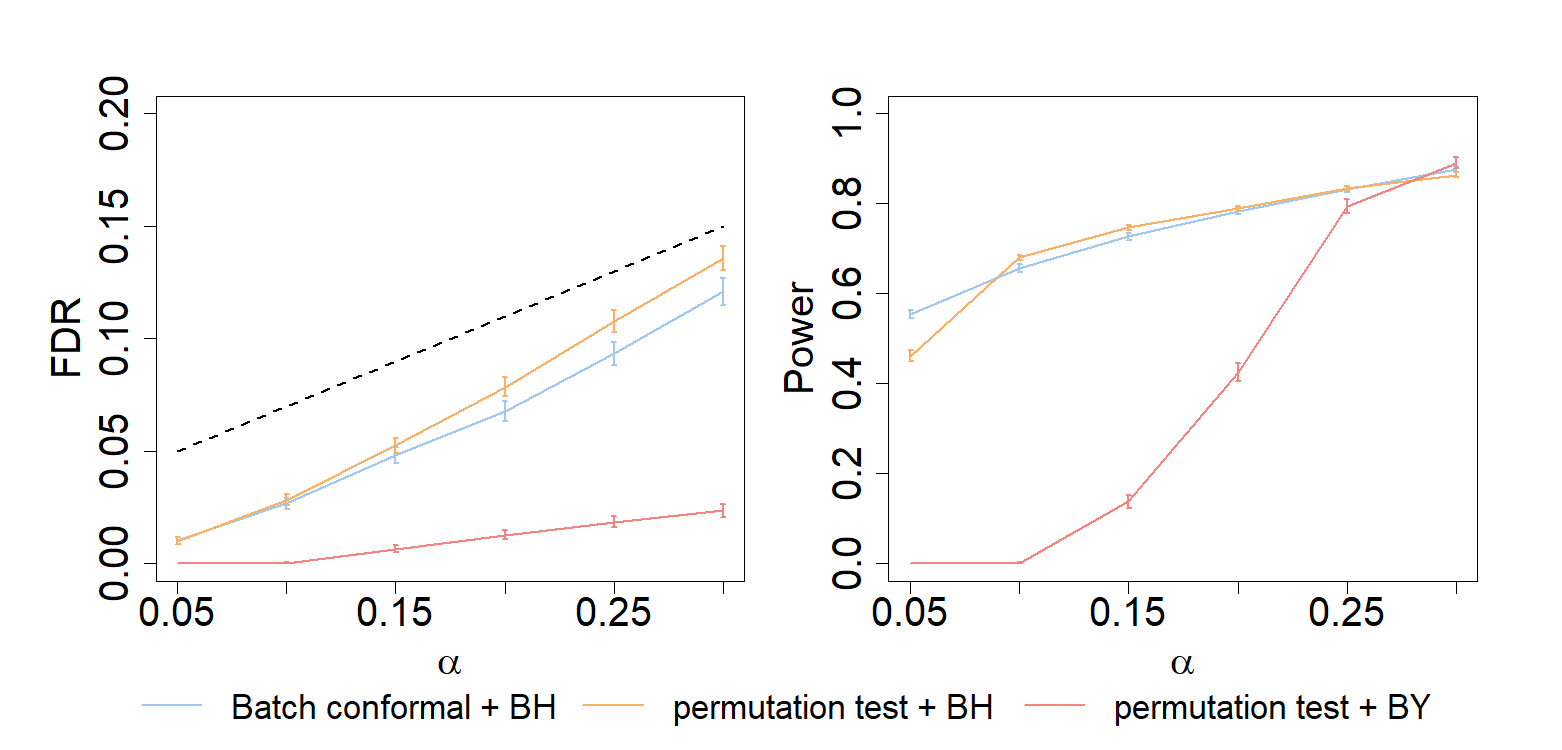}
    \caption{False discovery rate and power of the proposed procedure and the permutation-test-based methods.}
    \label{fig:mult_perm}
\end{figure}

\subsection{Two sample test under shift in the tail}

We provide additional experimental results that supplement the experiment in Section~\ref{ts-de}, to illustrate the usefulness of the batch conformal p-value in the setting where the shift occurs in the tail. We generate two samples from the following distributions:
\begin{align*}
(1)\; X \sim \mathcal{N}(0,1), \qquad 
(2)\; X \sim 0.9 \cdot \mathcal{N}(0,1) + 0.1 \cdot \delta_5,
\end{align*}
where $\delta_5$ denotes the point mass at $5$. In addition to the methods included in the experiment in Section~\ref{ts-de}---with $q=0.9$ for the batch conformal p-value and the permutation test---, we further compare with methods from~\cite{gretton2012kernel}, which introduced the kernel two-sample test based on the maximum mean discrepancy (MMD) statistic, and from~\cite{ramdas2017wasserstein}, which proposed a method based on the Wasserstein distance. We explore small-sample settings where distribution-free methods are typically preferred: (i) $n = m = 30$, and (ii) $n = 50$, $m = 20$, where $n$ and $m$ denote the sample sizes of the data from distributions (1) and (2), respectively.The results are shown in Figure~\ref{fig:two_sample_supp}, illustrating that the batch conformal p-value, with an appropriate choice of the quantile-based test statistic, achieves higher power compared to methods based on MMD or the Wasserstein distance.

\begin{figure}[ht]
    \centering
    \includegraphics[width=0.8\textwidth]{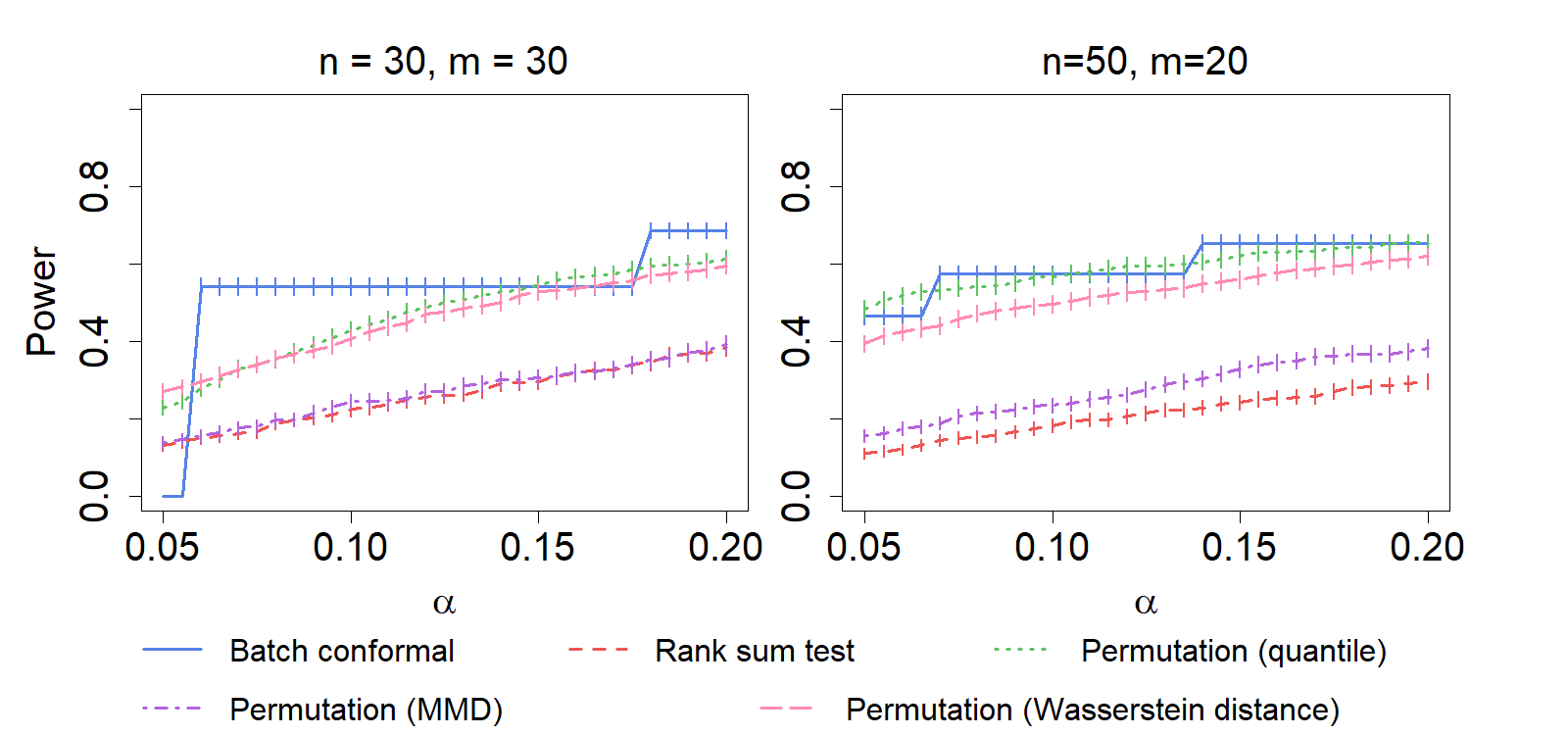}
    \caption{Power of different two sample testing procedures, under two settings of sample sizes.}
    \label{fig:two_sample_supp}
\end{figure}

\section{Proofs of theorems}\label{sec:proofs}

\subsection{Supporting lemmas}

We first introduce supporting lemmas for the validity of batch conformal $p$-values and Theorem~\ref{thm:prds}, establishing properties of the distribution of the ranks of test scores.

\begin{lemma}[\citet{lee2024batch}]\label{lem:ranks_1}
For positive integers $n, m \in \mathbb{N}$, let $S_1, S_2, \dots, S_{n+m} \sim P_S$ be exchangeable random variables, which are distinct almost surely. Let $\bS_{\uparrow} = (S_{(1)}, S_{(2)}, \dots, S_{(n+m)})$ be the vector of their order statistics. For $j \in [m]$, let $T_j = \sum_{k=1}^{n+m} \One{S_{(k)} \leq S_{n+j}}$ denote the rank of $S_{n+j}$ among the $n+m$ scores in increasing order, and let $R_1 < \ldots < R_m$ denote the ordered ranks obtained by rearranging $T_1,\ldots, T_m$. 
Then, 
\begin{equation}\label{eqn:rank_dist}
    (R_1,R_2,\ldots,R_m) \mid \bS_{\uparrow} \sim \textnormal{Unif}(\{(r_1,\ldots,r_m) : 1 \leq r_1 < \ldots < r_m \leq n+m\}).
\end{equation}
\end{lemma}

We omit the proof of Lemma~\ref{lem:ranks_1} as it is a direct consequence of the exchangeability condition and is also established in~\citet{lee2024batch}.\footnote{\citet{lee2024batch} derives that the marginal distribution of the vector $(R_1, \ldots, R_m)$ is uniform from the exchangeability condition. However, since conditioning on $\bS_{\uparrow}$---or equivalently, conditioning on the set of scores---does not break the exchangeability, the same steps hold for the conditional distribution of $(R_1, \ldots, R_m)$ given $\bS_{\uparrow}$.}

\begin{lemma}[Stochastic order relations between ranks]\label{lem:ranks_2}
If the random variables $R_1,\ldots,R_m$ follow the distribution~\eqref{eqn:rank_dist}, then:
\begin{enumerate}
    \item For any $t \in [m]$, $(R_1, \ldots, R_{t-1}) \indep (R_{t+1}, \ldots, R_m) \mid (R_t, \bS_{\uparrow})$.
    \item For any $t \in [m]$, $r \in [n+m]$ and $1 \le q \le q' \le n+m$, the following inequalities hold:
    \begin{align*}
        &\PPst{R_{t-1} \leq r}{R_t = q, \bS_{\uparrow}} \geq \PPst{R_{t-1} \leq r}{R_t = q', \bS_{\uparrow}}.\\
        &\PPst{R_{t+1} \geq r}{R_t = q, \bS_{\uparrow}} \leq \PPst{R_{t+1} \geq r}{R_t = q', \bS_{\uparrow}}.
    \end{align*}
    
\end{enumerate}
\end{lemma}


\begin{proof}[Proof of Lemma~\ref{lem:ranks_2}]
\noindent \textbf{Proof of the first claim. }
 Since the ranks $R_1,\ldots,R_m$ are independent of $\bS_{\uparrow}$ by Lemma~\ref{lem:ranks_1}, it is sufficient to prove that $(R_1,\ldots,R_{t-1})$ and $(R_{t+1}, \ldots, R_m)$ are conditionally independent given $R_t$.
 Let $L = |\{(r_1,\ldots,r_m) : 1 \leq r_1 < \ldots < r_m \leq n+m\}| = \binom{n+m}{m}$. Next, fix any $1 \leq r_1 < \ldots < r_m \leq n+m$, and compute
\begin{align*}
    &\PPst{R_1 = r_1, \ldots, R_{t-1} = r_{t-1}}{R_t = r_t} = \frac{\PP{R_1 = r_1, \ldots, R_{t-1}=r_{t-1}, R_t = r_t}}{\PP{R_t = r_t}}\\
    &= \frac{\sum_{r_t < v_{t+1} < \ldots < v_m \leq n+m} \PP{R_1 = r_1, \ldots, R_{t-1} = r_{t-1}, R_t = r_t, R_{t+1} = v_{t+1}, \ldots, R_m = v_m}}{\sum_{1 < u_1 < \ldots < u_{t-1} < r_t < v_{t+1} < \ldots < v_m \leq n+m} \PP{R_1 = u_1, \ldots, R_{t-1} = u_{t-1}, R_t = r_t, R_{t+1} = v_{t+1}, \ldots, R_m = v_m}}\\
    &= \frac{\frac{1}{L} \cdot \left|\left\{(v_{t+1},\ldots,v_m) : r_t < v_{t+1} < \ldots < v_m \leq n+m\right\}\right|}{\frac{1}{L} \cdot \left|\left\{(u_1, \ldots, u_{t-1}, v_{t+1},\ldots,v_m) : 1 \leq u_1 < \ldots < u_{t-1} < r_t < v_{t+1} < \ldots < v_m \leq n+m\right\}\right|}\\
    &= \frac{\binom{n+m-r_t}{m-t}}{\binom{r_t-1}{t-1} \cdot\binom{n+m-r_t}{m-t}} = \frac{1}{\binom{r_t-1}{t-1}}.
\end{align*}
Similarly, we have $\PPst{R_{t+1} = r_{t+1}, \ldots, R_m = r_m}{R_t = r_t} = 1/\binom{n+m-r_t}{m-t}$. 
Observe that the calculation 
above also reveals that 
$\PP{R_t = r_t} = \frac{1}{L} \cdot \binom{r_t-1}{t-1} \cdot\binom{n+m-r_t}{m-t}$. 
Thus, we have
\begin{multline*}
    \PPst{R_1 = r_1, \ldots, R_{t-1} = r_{t-1}, R_{t+1} = r_{t+1}, \ldots, R_m = r_m}{R_t = r_t}\\
    = \frac{\PP{R_1 = r_1, \ldots, R_{t-1} = r_{t-1}, R_t = r_t, R_{t+1} = r_{t+1}, \ldots, R_m = r_m}}{\PP{R_t = r_t}} = \frac{\frac{1}{L}}{\frac{1}{L} \cdot \binom{r_t-1}{t-1} \cdot\binom{n+m-r_t}{m-t}}\\
    = \PPst{R_1 = r_1, \ldots, R_{t-1} = r_{t-1}}{R_t = r_t}\cdot \PPst{R_{t+1} = r_{t+1}, \ldots, R_m = r_m}{R_t = r_t}.
\end{multline*}
The equality above holds for any sequence $r_1 < \ldots < r_m$, thereby proving the claim.\\

\noindent \textbf{Proof of the second claim. }
Fix $t\ge 2$.
Define $f_q(l) = \PP{R_{t-1} = l, R_t = q}$ for $1\le l < q \le n+m$.
Using similar calculations as before, we find that $f_q(l) = \frac{1}{L} \cdot \binom{l-1}{t-2}\cdot \binom{m+n-q}{m-t}$.
Therefore,
\begin{multline*}
    \PPst{R_{t-1} \leq r}{R_t = q} = \frac{\PP{R_{t-1} \leq r, R_t = q}}{\PP{R_t = q}} = \frac{\sum_{l=1}^r f_q(l)}{\sum_{l=1}^{q-1} f_q(l)} = \frac{\sum_{l=1}^r \binom{l-1}{t-2}}{\sum_{l=1}^{q-1} \binom{l-1}{t-2}} \geq \frac{\sum_{l=1}^r \binom{l-1}{t-2}}{\sum_{l=1}^{q'-1} \binom{l-1}{t-2}}\\ = \PPst{R_{t-1} \leq r}{R_t = q'}.
\end{multline*}
The second inequality can be derived similarly.

\end{proof}

\subsection{Proof of Marginal Validity of Batch Conformal P-values}
\label{pfprop:pval}
For simplicity, we use the 
notations 
for two-sample testing (Section \ref{sec:two_sample}),
i.e., write 
$P^{(k)}=Q$, $n_k=m$, and 
$(S_{1}^{(k)}, \ldots, S_{n_k}^{(k)}):=(S_{n+1}, \ldots, S_{n+m})$.
We first consider the case where the scores are all distinct almost surely. It is sufficient to prove that
\[p = \sum_{i=1}^n w_i \cdot \One{S_{(\eta)}^\text{test} \leq S_{(i)}} + w_{n+1},\qquad\text{ where } w_i = \frac{\binom{\eta+i-2}{\eta-1} \binom{n+m-\eta-i+1}{m-\eta}}{\binom{n+m}{m}}, i\in [n]\]
is super-uniform under $H_0$, for any $n,m$, and $0 \leq \eta \leq m$. 
Define $R_1 < \ldots < R_m$ as in Lemma~\ref{lem:ranks_1}. 
By the definition of $R_\eta$, there are $R_\eta - \eta$ number of reference scores that are smaller than $S_{(\eta)}^\text{test}$. Consequently, for any $i \in [n]$, the following equivalence holds:
\begin{equation}\label{eqn:s_r}
    S_{(\eta)}^\text{test} \leq S_{(i)} \iff R_\eta - \eta +1 \leq i.
\end{equation}
Now we find the distribution of $R_\eta - \eta + 1$. By Lemma~\ref{lem:ranks_1}, for any $i \in [n+1]$,
\begin{align*}
    \PP{R_\eta - \eta +1 = i}
    &= \frac{1}{\binom{n+m}{m}}\cdot\big|\{(r_1,\ldots,r_m) : 1 \leq r_1 < \ldots < r_m \leq n+m, r_\eta = \eta+i-1\}\big|\\
    &= \frac{1}{\binom{n+m}{m}}\cdot\binom{\eta+i-2}{\eta-1} \binom{n+m-\eta-i+1}{m-\eta} = w_i.
\end{align*}
It follows that the cumulative distribution function of $R_\eta - \eta +1$ is $x\mapsto F_{R_\eta - \eta+1}(x) = \sum_{i=1}^{n+1} w_i\One{i \leq x}$. Therefore, for any $\alpha \in (0,1)$,
\begin{multline}\label{eqn:ineq_1}
    \PP{p \leq \alpha} = \PP{\sum_{i=1}^n w_i \One{S_{(\eta)}^\text{test} \leq S_{(i)}} + w_{n+1} \leq \alpha} = \PP{\sum_{i=1}^{n+1} w_i \One{R_\eta - \eta +1 \leq i} \leq \alpha}\\
    = \PP{1-F_{R_\eta - \eta+1}(R_\eta - \eta) \leq \alpha} = \PP{F_{R_\eta - \eta+1}(R_\eta - \eta) \geq 1-\alpha}\\
    \leq \PP{F_{R_\eta - \eta+1}(R_\eta - \eta+1) > 1-\alpha} \leq \alpha,
\end{multline}
as desired.

Now, suppose ties exist among the scores, and we apply uniform tie-breaking to determine the ranks. This ensures that the uniform distribution of the rank vector, i.e., the result of Lemma~\ref{lem:ranks_1}, still holds, but we only obtain  
\[S_{(\eta)}^\text{test} \leq S_{(i)} \Longleftarrow R_\eta - \eta +1 \leq i.\]  
instead of relation~\eqref{eqn:s_r}. Nevertheless, this still leads to inequality~\eqref{eqn:ineq_1}, provided that the second equality is replaced with an inequality.

\subsection{Proof of Theorem~\ref{thm:prds}}
\label{pfthm:prds}

 We need to show that for each $k$ such that $H_k$ is true,
 \(x\mapsto \PPst{(p_1,\ldots,p_K) \in A}{p_k = x}\)
 is nondecreasing over its domain of definition (i.e., over the set of $x$s where the probability is well-defined), 
  for any increasing set $A \subset \R^{K}$. Without loss of generality, suppose $H_1$ is true, and fix any increasing set $A \subset \R^{K}$ and $x \in (0,1)$. 
 
For conciseness, we denote the comparison 
sample size as  $m := n_1$, and 
 the comparison scores as $S_{n+j} := S_j^{(1)}$, $j\in [m]$.
Let $\bS_{(i)}$, $i\in[n+m]$ be the $i$-th order statistic of $S_1,\ldots, S_{n+m}$---we write $\bS_{(i)}$ to distinguish it from
the $i$-th order statistic  $S_{(i)}$ of $S_1,\ldots,S_n$. 
Let $\bS_{\uparrow} = (\bS_{(1)}, \ldots, \bS_{(n+m)})$ and $S_{\uparrow} = (S_{(1)}, \ldots, S_{(n)})$.
Let us define $R_1, R_2,\ldots, R_m$ as in Lemma~\ref{lem:ranks_1}. 
For simplicity, write $\eta := \eta_1$. 

By definition, the batch conformal $p$-value $p_1$ from \eqref{eqn:p_val_k} can be written as
 \begin{equation}\label{eqn:p_1}
      p_1 = \sum_{i=1}^n w_i \cdot \mathbf{1}\{S_{(\eta)}^{(1)} \leq S_{(i)}\} + w_{n+1} 
 = \sum_{i : S_{(i)} \geq \bar{S}_{(R_\eta)}} w_i + w_{n+1} 
 = \sum_{i \geq R_\eta - \eta + 1} w_i + w_{n+1},
 \end{equation} 
 where $w_i = \binom{i+\eta-2}{\eta-1}\binom{n+m-i-\eta+1}{m-\eta}/\binom{n+m}{m}$ for $i\in [n+1]$. 
 Therefore, the event $p_1 = x$ can be written in the form $R_\eta = r$ for some $r = r(x) \in [n+m]$. 
 Further, by inspection, $p_1$ is a decreasing function of $R_\eta$.
Putting everything together, our task reduces to proving that
\(r\mapsto\PPst{(p_1,\ldots,p_K) \in A}{R_\eta=r}\)
is non-increasing in $r$.

Recalling $S_{\uparrow} = (S_{(1)}, \ldots, S_{(n)})$,
observe that
 \begin{multline}\label{eqn:prds_1}
     \PPst{(p_1,\ldots,p_K) \in A}{R_\eta=r} = \EEst{\PPst{(p_1,\ldots,p_K) \in A}{R_\eta=r, S_{\uparrow}}}{R_\eta = r}\\
     \EEst{\PPst{(p_1^r,p_2,\ldots,p_K) \in A}{R_\eta=r, S_{\uparrow}}}{R_\eta = r} = \Ep{S_{\uparrow} \sim P_{S_{\uparrow} \mid R_\eta = r}}{\PPst{(p_1^r,p_2,\ldots,p_K) \in A}{S_{\uparrow}}},
 \end{multline}
 where $p_1^r$ is defined as in~\eqref{eqn:p_1} but with $R_\eta$ replaced with $r$. Now we examine the distribution of $S_{\uparrow}$ given $R_\eta = r$. By definition of $R_1,\ldots, R_m$ from Lemma \ref{lem:ranks_1} 
as the ordered ranks of 
$S_{n+1}, \ldots, S_{n+m}$ among $\bS_{\uparrow} = (\bar S_{(1)}, \bar S_{(2)}, \dots, \bar S_{(n+m)})$, 
we can write the remaining ordered ranks 
$S_{\uparrow} = (S_{(1)}, \ldots, S_{(n)})$ as 
$\bS_{\uparrow}$ with the coordinates $\bS_{(R_j)}$, $j\in [m]$, removed: 
\begin{equation}\label{eqn:S_0}
    S_{\uparrow} = (\bS_{(1)}, \ldots, \bS_{(R_1-1)}, \bS_{(R_1+1)}, \ldots, \bS_{(R_2-1)}, \bS_{(R_2+1)},\ldots, \bS_{(R_m-1)}, \bS_{(R_m+1)}, \ldots ,\bS_{(n+m)}).
\end{equation}
Thus, $S_{\uparrow}$ is fully determined by $\bS_{\uparrow}$ and $(R_1,\ldots, R_m)$. Therefore, the conditional distribution of $S_{\uparrow}$ given $\bS_{\uparrow}$ and $R_\eta = r$ is determined by the conditional distribution of $(R_1,\ldots, R_m)$ given $\bS_{\uparrow}$ and $R_\eta = r$.
We will leverage this result in the proof.

Next, we will show the following lemma, which allows to represent conditional distributions
of $S_{\uparrow} \mid R_\eta$ for two different values of $R_\eta$ as a jointly distributed pair of random vectors decreasing in $R_\eta$.
Below, for two random objects $U,V$, $U \stackrel{d}{=} V$ denotes that they have the same distribution.

\begin{lemma}[Almost sure representation of conditional distributions
of $S_{\uparrow} \mid R_\eta$ for two values of $R_\eta$]
For any $r_1,r_2\in [n+m]$, 
$r_1 < r_2$, there exist random vectors $V_1$ and $V_2$ over $\R^n$ such that the following two conditions hold with $S_{\uparrow} = (S_{(1)}, \ldots, S_{(n)})$:
\[\textnormal{ (1) } V_1 \stackrel{d}{=} (S_{\uparrow} \mid R_\eta = r_1) 
\text{ and } 
V_2 \stackrel{d}{=} (S_{\uparrow} \mid R_\eta = r_2) \qquad\qquad 
\textnormal{ (2) } V_1 \succeq V_2 \textnormal{ almost surely}.\]
\end{lemma}
\begin{proof}
We first draw $\bS_{\uparrow} = (\bS_{(1)}, \ldots, \bS_{(n+m)})$, and then show that there exist random vectors $(\tilde{R}_1,\ldots, \tilde{R}_m)$ and $(\hat{R}_1,\ldots, \hat{R}_m)$ such that the following conditions hold:
\begin{equation}\label{eqn:rank_alt}
\begin{split}
    &(\tilde{R}_1,\ldots, \tilde{R}_m) \mid \bS_{\uparrow} \stackrel{d}{=} (R_1,\ldots,R_m) \mid R_\eta = r_1, \bS_{\uparrow}\\
    &(\hat{R}_1,\ldots, \hat{R}_m) \mid \bS_{\uparrow} \stackrel{d}{=} (R_1,\ldots,R_m) \mid R_\eta = r_2, \bS_{\uparrow}\\
    &(\tilde{R}_1,\ldots, \tilde{R}_m) \preceq (\hat{R}_1,\ldots, \hat{R}_m) \textnormal{ almost surely}.
    \end{split}
\end{equation}
Then we will argue that the claim of the lemma follows from the observation~\eqref{eqn:S_0}, and in particular since $S_{\uparrow}$ is determined by $\bS_{\uparrow}$ and $(R_1,\ldots, R_m)$. 

Now, 
we condition on the event $\bS_{\uparrow} = \bs_{\uparrow}$ for some realization $\bs_{\uparrow}$, and construct $\tilde{R}_j$s and $\hat{R}_j$s by induction on $m$.
First, we observe that for any $t \in [m]$, $t\ge 2$, and
$q,q' \in [m+n]$,
$q \le q'$, there exist random variables $\tilde{R}_{t-1,q,q'}$ and $\hat{R}_{t-1,q,q'}$ such that the following conditions hold. 
\begin{align*}
&\tilde{R}_{t-1,q,q'} \mid  \bS_{\uparrow} = \bs_{\uparrow} \;\stackrel{d}{=}\; R_{t-1} \mid R_t = q, \bS_{\uparrow} = \bs_{\uparrow},\\
&\hat{R}_{t-1,q,q'} \mid  \bS_{\uparrow} = \bs_{\uparrow} \;\stackrel{d}{=}\; R_{t-1} \mid R_t = q', \bS_{\uparrow} = \bs_{\uparrow},\\
&\tilde{R}_{t-1,q,q'} \leq \hat{R}_{t-1,q,q'} \quad\text{almost surely}.
\end{align*}
This is a direct consequence of the second claim of Lemma~\ref{lem:ranks_2} and Strassen's theorem~\citep{strassen1965existence}, which implies that for any random variables $X$ and $Y$ such that $X$ stochastically dominates $Y$, there exist random variables $X'$ and $Y'$ such that $X' \stackrel{d}{=} X$, $Y' \stackrel{d}{=} Y$, and $X' \geq Y'$ almost surely. Similarly, for any $t$ and $q \le q'$, there exist random variables $\tilde{R}_{t+1,q,q'}'$ and $\hat{R}_{t+1,q,q'}'$ such that
\begin{align*}
&\tilde{R}_{t+1,q,q'}' \mid \bS_{\uparrow} = \bs_{\uparrow} \;\stackrel{d}{=}\; R_{t+1} \mid R_t = q, \bS_{\uparrow} = \bs_{\uparrow},\\
&\hat{R}_{t+1,q,q'}' \mid \bS_{\uparrow} = \bs_{\uparrow} \;\stackrel{d}{=}\; R_{t+1} \mid R_t = q', \bS_{\uparrow} = \bs_{\uparrow},\\
&\tilde{R}_{t+1,q,q'}' \leq \hat{R}_{t+1,q,q'}'\quad\text{almost surely}
\end{align*}
hold. Moreover, we can construct the 
rank pairs $(\tilde{R}_{t,q,q'},\hat{R}_{t,q,q'})$s and $(\tilde{R}_{t,q,q'}',\hat{R}_{t,q,q'}')$s 
to be jointly independent and also independent of the actual ranks $R_1,\ldots,R_m$ (conditional on $\bS_{\uparrow} = \bs_{\uparrow}$)---by applying Strassen's theorem separately for the two cases.

Next, define the $\tilde{R}_j$s and $\hat{R}_j$s as follows.

\begin{enumerate}
\item Set $\tilde{R}_\eta = r_1$ and $\hat{R}_\eta = r_2$.

\item For $j=\eta-1, \ldots, 2,1$, define $\tilde{R}_{j} = \tilde{R}_{j,\tilde{R}_{j+1}, \hat{R}_{j+1}}$ and $\hat{R}_{j} = \hat{R}_{j,\tilde{R}_{j+1}, \hat{R}_{j+1}}$.

\item For $j=\eta+1, \eta+2, \ldots, m$, define $\tilde{R}_{j} = \tilde{R}_{j,\tilde{R}_{j-1}, \hat{R}_{j-1}}'$ and $\hat{R}_{j} = \hat{R}_{j,\tilde{R}_{j-1}, \hat{R}_{j-1}}'$.
\end{enumerate}
The above construction is well defined, since for each $j\in [m+n]$, $\tilde{R}_j \leq \hat{R}_{j}$ holds by induction.
Then, observe that for each $j < \eta$,
\begin{align*}
    &\tilde{R}_j \mid \tilde{R}_{j+1} = q, \hat{R}_{j+1} = q', \tilde{R}_{j+2}, \tilde{R}_{j+3}, \ldots, \tilde{R}_\eta, \bS_{\uparrow} = \bs_{\uparrow}\\
    &\stackrel{d}{=} \tilde{R}_{j,q,q'} \mid \tilde{R}_{j+1} = q, \hat{R}_{j+1} = q', \tilde{R}_{j+2}, \tilde{R}_{j+3}, \ldots, \tilde{R}_\eta, \bS_{\uparrow} = \bs_{\uparrow}\qquad\textnormal{by definition of $\tilde{R}_j$}\\
    &\stackrel{d}{=} \tilde{R}_{j,q,q'} \mid \bS_{\uparrow} = \bs_{\uparrow} \qquad\textnormal{by the independence of the constructed rank pairs}\\
    &\stackrel{d}{=} R_j \mid R_{j+1} = q, \bS_{\uparrow} = \bs_{\uparrow} \qquad\textnormal{by the construction of $\tilde{R}_{j,q,q'}$}.
\end{align*}
Since the above holds for any $q\in [m+n]$ and the final distribution does not depend on $q'\in [m+n]$, this implies
\begin{align*}
    (\tilde{R}_j \mid \tilde{R}_{j+1}, \tilde{R}_{j+2}, \ldots, \tilde{R}_\eta, \bS_{\uparrow}  
    =  \bs_{\uparrow} )
    &\stackrel{d}{=} R_j \mid R_{j+1}, \bS_{\uparrow} = \bs_{\uparrow}\\
    &\stackrel{d}{=} R_j \mid R_{j+1}, R_{j+2}, \ldots, R_\eta, \bS_{\uparrow} = \bs_{\uparrow},
\end{align*}
where the second equality follows by the first claim of Lemma~\ref{lem:ranks_2}. 
By an analogous argument, we have
\[(\hat{R}_j \mid \hat{R}_{j+1}, \hat{R}_{j+2}, \ldots, \hat{R}_\eta, \bS_{\uparrow} =  \bs_{\uparrow}) 
\stackrel{d}{=} R_j \mid R_{j+1}, R_{j+2}, \ldots, R_\eta, \bS_{\uparrow} = \bs_{\uparrow}\]
for any $j < \eta$. Next, for each $j > \eta$, we have
by similar arguments that
\begin{align*}
&(\tilde{R}_j \mid \tilde{R}_{j-1}, \tilde{R}_{j-2}, \ldots, \tilde{R}_\eta, \ldots, \tilde{R}_{1}, \bS_{\uparrow} =  \bs_{\uparrow}) 
\stackrel{d}{=} R_j \mid R_{j-1}, R_{j-2}, \ldots, R_\eta, \ldots, R_1, \bS_{\uparrow} = \bs_{\uparrow},\\
&(\hat{R}_j \mid \hat{R}_{j-1}, \hat{R}_{j-2}, \ldots, \hat{R}_\eta, \ldots, \hat{R}_{1}, \bS_{\uparrow} =  \bs_{\uparrow}) 
\stackrel{d}{=} R_j \mid R_{j-1}, R_{j-2}, \ldots, R_\eta, \ldots, R_1, \bS_{\uparrow} = \bs_{\uparrow}.
\end{align*}
Therefore, putting everything together, the random vectors $(\tilde{R}_1,\ldots, \tilde{R}_m)$ and $(\hat{R}_1,\ldots, \hat{R}_m)$ satisfy the conditions in~\eqref{eqn:rank_alt}. Now we set $V_1$ and $V_2$ as
\begin{align*}
    V_1 = (\bS_{(1)}, \ldots, \bS_{(\tilde{R}_1-1)}, \bS_{(\tilde{R}_1+1)}, \ldots, \bS_{(\tilde{R}_2-1)}, \bS_{(\tilde{R}_2+1)},\ldots, \bS_{(\tilde{R}_m-1)}, \bS_{(\tilde{R}_m+1)}, \ldots ,\bS_{(n+m)}),\\
    V_2 = (\bS_{(1)}, \ldots, \bS_{(\hat{R}_1-1)}, \bS_{(\hat{R}_1+1)}, \ldots, \bS_{(\hat{R}_2-1)}, \bS_{(\hat{R}_2+1)},\ldots, \bS_{(\hat{R}_m-1)}, \bS_{(\hat{R}_m+1)}, \ldots ,\bS_{(n+m)}).
\end{align*}
Then $V_1 \succeq V_2$ holds almost surely, since $(\tilde{R}_1,\ldots, \tilde{R}_m) \preceq (\hat{R}_1,\ldots, \hat{R}_m)$ holds with probability one. 

Furthermore, from~\eqref{eqn:rank_alt}, we have 
$V_1 \mid \bS_{\uparrow} \stackrel{d}{=} (S_{\uparrow} \mid R_\eta = r_1, \bS_{\uparrow})
\text{ and } 
V_2 \mid \bS_{\uparrow} \stackrel{d}{=} (S_{\uparrow} \mid R_\eta = r_2, \bS_{\uparrow})$. 
By marginalizing over $\bS_{\uparrow}$, we also obtain 
$V_1 \stackrel{d}{=} (S_{\uparrow} \mid R_\eta = r_1)
\text{ and }
V_2 \stackrel{d}{=} (S_{\uparrow} \mid R_\eta = r_2)$.
    \end{proof}

Now, for each $k \in [K]$, let $p_{k}^{(1)}$ and $p_k^{(2)}$ be defined as in~\eqref{eqn:p_val_k}, with $S_{\uparrow}$ replaced by $V_1$ and $V_2$, respectively. 
Now, $p_j^{(1)} \geq p_j^{(2)}$ holds almost surely for all $j \in [K]$ by the definition~\eqref{eqn:p_val_k} and the fact that $V_1 \succeq V_2$ holds almost surely. 
Since $A$ is an increasing set, from \eqref{eqn:prds_1} and the above results, we thus have that
\begin{multline*}
    \PPst{(p_1,\ldots,p_K) \in A}{R_\eta=r_1} = \Ep{S_{\uparrow} \sim P_{S_{\uparrow} \mid R_\eta = r_1}}{\PPst{(p_1^{r_1},p_2,\ldots,p_K) \in A}{S_{\uparrow}}}\\
    = \EE{\PPst{(p_1^{r_1},p_2^{(1)},\ldots,p_{k}^{(1)}) \in A}{V_1}}  = \PP{(p_1^{r_1},p_2^{(1)},\ldots,p_{k}^{(1)}) \in A} \geq \PP{(p_1^{r_2},p_2^{(2)},\ldots,p_K^{(2)}) \in A}\\
    = \PPst{(p_1,\ldots,p_K) \in A}{R_\eta=r_2},
\end{multline*}
as desired. Note that $p_1^{r_1} \geq p_1^{r_2}$ holds by definition.

The second claim of Theorem~\ref{thm:prds} follows from the result of~\citet[Theorem 1.2]{benjamini2001control}.

\subsection{Proof of Theorem~\ref{thm:two_quantiles}}
\label{pfthm:two_quantiles}
Let us define $R_1 < \ldots < R_m$ as in the proof of the validity of batch conformal $p$-values.
By the observation in~\eqref{eqn:s_r}, we have the following equivalence:

\begin{align*}
    S_{(\eta_1)}^\text{test} \leq S_{(\tilde{\eta}_1+t)}, S_{(\eta_2)}^\text{test} \leq S_{(\tilde{\eta}_2+t)} &\iff R_{\eta_1} - \eta_1+1 \leq \tilde{\eta}_1+t \text{ and } R_{\eta_2} - \eta_2+1 \leq \tilde{\eta}_2+t \\
    &\iff T \leq t, \textnormal{ where } T = \max\{R_{\eta_1} - \eta_1 - \tilde{\eta}_1+1, R_{\eta_2} - \eta_2 - \tilde{\eta}_2+1\}.
\end{align*}
Next, we compute the probability mass function of $T$. For $-\tilde{\eta}_1+1 \leq t \leq n-\tilde{\eta}_2+1$,
\begin{align*}
    \PP{T = t} &= \PP{\max\{R_{\eta_1} - \eta_1 - \tilde{\eta}_1+1, R_{\eta_2} - \eta_2 - \tilde{\eta}_2+1\} = t}\\
    &=  \sum_{\substack{t_1, t_2 \,:\, \max\{t_1,t_2\} = t \\ 1 \leq \eta_1+\tilde{\eta}_1+t_1 + 1 < \eta_2+\tilde{\eta}_2+t_2 + 1 \leq n+m}} \PP{R_{\eta_1} - \eta_1 - \tilde{\eta}_1+1 = t_1, R_{\eta_2} - \eta_2 - \tilde{\eta}_2+1 = t_2}\\
    &= \sum_{\substack{t_1, t_2 \,:\, \max\{t_1,t_2\} = t \\ 1 \leq \eta_1+\tilde{\eta}_1+t_1 + 1 < \eta_2+\tilde{\eta}_2+t_2 + 1 \leq n+m}} \PP{R_{\eta_1} = t_1 + \eta_1 + \tilde{\eta}_1-1, R_{\eta_2} = t_2 + \eta_2 + \tilde{\eta}_2-1}\\
    &= \sum_{\substack{t_1, t_2 \,:\, \max\{t_1,t_2\} = t \\ 1 \leq \eta_1+\tilde{\eta}_1+t_1 + 1 < \eta_2+\tilde{\eta}_2+t_2 + 1 \leq n+m}}\frac{\binom{\eta_1+\tilde{\eta}_1+t_1-2}{\eta_1-1}\binom{\eta_2+\tilde{\eta}_2+t_2-\eta_1+\tilde{\eta}_1+t_1-1}{\eta_2-\eta_1-1}\binom{n+m-\eta_2-\tilde{\eta}_2-t_2+1}{m-\eta_2}}{\binom{n+m}{m}} = w_t.
\end{align*}
In the last line, we have used the definition of $w_t$ from \eqref{eqn:p_two_quantiles}.
Therefore, the cumulative distribution function of $T$ is given by $\tau\mapsto F_T(\tau) = \sum_{t = -\tilde{\eta}_1+1}^{n-\tilde{\eta}_2+1} w_t \One{t \leq \tau}$, and thus we have that
\begin{align*}
    \PP{p \leq \alpha} &= \PP{\sum_{t = -\tilde{\eta}_1+1}^{n-\tilde{\eta}_2+1} w_t \cdot \One{S_{(\eta_1)}^\text{test} \leq S_{(\tilde{\eta}_1+t)}, S_{(\eta_2)}^\text{test} \leq S_{(\tilde{\eta}_2+t)}} \leq \alpha}\\
    &= \PP{\sum_{t = -\tilde{\eta}_1+1}^{n-\tilde{\eta}_2+1} w_t \One{T \leq t} \leq \alpha} = \PP{F_T(T) \leq \alpha} \leq \alpha
\end{align*}
holds for any $\alpha \in (0,1)$.
This finishes the proof.

\section{Positive regression dependence of one-sided two sample normal p-values}\label{sec:z_prds}

Here, we show that the oracle
one-sided two sample $p$-values
\( (p_k)_{1 \leq k \leq K} \), as defined in~\eqref{eqn:p_val_z}, satisfy positive regression dependence. Consequently, applying the Benjamini-Hochberg procedure to these $p$-values ensures valid false discovery rate control.

\begin{proposition}\label{prop:prds_z}
    Suppose the null distribution is $\mathcal{N}(0,\sigma^2)$. Then the $p$-values $p_1,p_2,\ldots,p_K$ defined as 
    \begin{equation}\label{eqn:p_val_z}
p_k = \Phi\left(\frac{\bar{X}_{\mathrm{ref}} - \bar{X}_k}{\sqrt{\frac{\sigma^2}{n} + \frac{\sigma^2}{n_k}}}\right),\quad k=1,2,\ldots,K,
\end{equation} are positive regression dependent on the set of nulls.
\end{proposition}

\subsection{Proof of Proposition~\ref{prop:prds_z}}
Fix an increasing set $A \subset \mathbb{R}^{K}$. 
We need to show that for each $k$ such that $H_k$ is true,  
as a function of $x \in (0,1)$,
$x\mapsto\PPst{(p_1, \ldots, p_K) \in A}
{p_k = x}$
is nondecreasing.
Without loss of generality, we assume that $H_1$ is true, and then prove the monotonicity of the above quantity for $k=1$. 
Let $m:=n_1$ for simplicity.

We first investigate the conditional distribution of $\bar{X}_{\mathrm{ref}}$ given $p_1 = x$.
By the definition of $p_1$, the event $p_1 = x$ is equivalent to $\bar{X}_{\mathrm{ref}} - \bar{X}_1 = z$, where $z = \Phi^{-1}(x)\cdot\sigma\sqrt{\frac{1}{n} + \frac{1}{m}}$. 
Since the reference and the comparison datapoints are i.i.d.~$\mathcal{N}(0, \sigma^2)$ under the null,  the conditional distribution of $\bar{X}_{\mathrm{ref}}$ given $\bar{X}_{\mathrm{ref}} - \bar{X}_1 = z$ is $\mathcal{N}\left(\frac{m}{n+m}z, \frac{\sigma^2}{n+m}\right)$. 
Now take any $x'\in(0,1)$, $x' > x$ and let $z' = \Phi^{-1}(x')\cdot\sigma\sqrt{\frac{1}{n} + \frac{1}{m}}$. Since $\Phi^{-1}$ is an increasing function, we have $z' > z$.
This implies that the distribution $\mathcal{N}\left(\frac{m}{n+m}z', \frac{\sigma^2}{n+m}\right)$ is stochastically larger than the distribution $\mathcal{N}\left(\frac{m}{n+m}z, \frac{\sigma^2}{n+m}\right)$.
Therefore, putting everything together, there exist random variables $V_1$ and $V_2$ such that the following conditions hold.
\begin{align*}
    (1)\; V \stackrel{d}{=} \bar{X}_{\mathrm{ref}} \mid p_1 = x,\quad (2)\; V' \stackrel{d}{=} \bar{X}_{\mathrm{ref}} \mid p_1 = x' \quad (3)\; V \leq V' \text{ almost surely}.
\end{align*}
Next, for each $k =2,3,\ldots,K$,
define $p_k^V$ as
$p_k^V = \Phi\left(\frac{V - \bar{X}_k}{\sqrt{\frac{\sigma^2}{n} + \frac{\sigma^2}{n_k}}}\right)$,
and define $p_k^{V'}$ similarly.
By condition (3) above, $p_k^V \leq p_k^{V'}$ holds almost surely, for each $k =2,3,\ldots,K$.
Therefore, we have
\begin{multline*}
    \PPst{(p_1, \ldots, p_K) \in A}{p_k = x} = \PP{(p_1^V, \ldots, p_K^V) \in A} \leq \PP{(p_1^{V'}, \ldots, p_K^{V'}) \in A}\\
    = \PPst{(p_1, \ldots, p_K) \in A}{p_k = x'},
\end{multline*}
where the first and the second equality apply conditions (1) and (2), respectively. 
This proves positive regression dependence.

\section{Experiments with CPS work hours data}
\label{cps}

In this section, we provide experimental results using the 
Current Population Survey (CPS)
dataset.\footnote{This dataset was accessed at \url{https://murraylax.org/datasets/cpshours.csv}.
The dataset has been analyzed at \url{https://murraylax.org/rtutorials/oneway-anova.html}. It includes data from more than 52,000 individuals over the age of 25 years that participated in the 2016 Current Population Survey.
}\footnote{Code to reproduce the experiments in this section is available at \url{https://github.com/yhoon31/batch_conformal}.} 
This experiment is an example of a more straightforward illustration---without additional steps required to apply the procedure, as was necessary in the HALT-C dataset, where we had to account for unobserved counterfactual variables.

The dataset consists of demographic information and education levels of various individuals, and the outcome of interest is their work hours. We form groups based on age, sex, race, and education level. Sex has two categories: “Male” and “Female”. Race has five categories: “White”, “Black”, “Asian/Pacific Islander”, “American Indian/Aleut/Eskimo”, and “Other”. Education level has four categories: “High school”, “Some college”, “Four-year degree”, and “Advanced degree”. Age is divided into four groups: “25–39”, “40–49”, “50–59”, and “60+”. Therefore, there are $2 \cdot 5 \cdot 4 \cdot 4 = 160$ possible groups, and we include 157 groups in the experiments, each of which has at least five datapoints.

The group with the largest sample size---Female/White/Some college/25–39---is set as the reference group. Although the original group contains 2,912 datapoints, we randomly select 100 to simulate a small-sample setting, where distinguishing between distributions is more difficult, and use these as the reference data. Similarly, for comparison groups with more than 50 datapoints, we randomly select 50 to use as the comparison data.

We run our procedure using the outcome variable---work hours---itself as the score, with three test statistics: the median and the quartiles, as in the experiment with the HALT-C dataset. This leads to the detection of groups whose overall work hours are significantly higher compared to the reference group. Table~\ref{tab:work_hour} shows the list of selected groups from the procedure at levels $\alpha = 0.01$ and $0.05$, demonstrating that there are larger differences in higher quantiles.

\begin{table}[htbp]
\centering
\resizebox{0.9\textwidth}{!}{%
\begin{tabular}{llll|cc|cc|cc}
\specialrule{0.6pt}{0pt}{0pt}
\toprule
 &  &  &  & \multicolumn{2}{c|}{$Q_1$} & \multicolumn{2}{c|}{$Q_2$} & \multicolumn{2}{c}{$Q_3$} \\
\midrule
Sex & Race & Education & Age & 0.01 & 0.05 & 0.01 & 0.05 & 0.01 & 0.05 \\
\midrule
\midrule
M & American Indian/Aleut/Eskimo & High School & 25--39 & -- & -- & -- & \checkmark & -- & -- \\
M & Other & High School & 50--59 & -- & \checkmark & -- & -- & -- & -- \\
M & White & High School & 25--39 & -- & -- & -- & \checkmark & -- & -- \\
M & White & High School & 40--49 & -- & -- & -- & \checkmark & -- & \checkmark \\
M & American Indian/Aleut/Eskimo & Some College & 25--39 & -- & -- & -- & \checkmark & -- & -- \\
F & Black & Some College & 40--49 & -- & -- & -- & \checkmark & -- & -- \\
M & Black & Some College & 25--39 & -- & -- & -- & -- & -- & \checkmark \\
M & Black & Some College & 40--49 & -- & -- & -- & -- & -- & \checkmark \\
M & Black & Some College & 50--59 & -- & -- & -- & \checkmark & -- & -- \\
M & Other & Some College & 25--39 & -- & -- & -- & \checkmark & -- & \checkmark \\
M & Other & Some College & 40--49 & -- & \checkmark & \checkmark & \checkmark & \checkmark & \checkmark \\
M & White & Some College & 25--39 & -- & \checkmark & \checkmark & \checkmark & -- & -- \\
M & White & Some College & 40--49 & -- & \checkmark & \checkmark & \checkmark & \checkmark & \checkmark \\
M & White & Some College & 50--59 & -- & -- & -- & -- & -- & \checkmark \\
M & White & Some College & 60+ & -- & -- & -- & -- & \checkmark & \checkmark \\
M & Asian/Pacific Islander & Four Year & 25--39 & -- & -- & -- & \checkmark & -- & -- \\
M & Asian/Pacific Islander & Four Year & 40--49 & -- & \checkmark & \checkmark & \checkmark & -- & \checkmark \\
M & Asian/Pacific Islander & Four Year & 50--59 & -- & \checkmark & -- & -- & -- & -- \\
M & Asian/Pacific Islander & Four Year & 60+ & -- & \checkmark & -- & -- & -- & -- \\
M & Black & Four Year & 40--49 & -- & -- & -- & \checkmark & -- & \checkmark \\
M & Black & Four Year & 50--59 & -- & -- & -- & -- & -- & \checkmark \\
M & Other & Four Year & 50--59 & -- & \checkmark & -- & \checkmark & -- & -- \\
F & White & Four Year & 40--49 & -- & -- & -- & \checkmark & -- & \checkmark \\
M & White & Four Year & 25--39 & -- & -- & -- & \checkmark & -- & \checkmark \\
M & White & Four Year & 40--49 & -- & \checkmark & \checkmark & \checkmark & \checkmark & \checkmark \\
M & White & Four Year & 50--59 & -- & \checkmark & \checkmark & \checkmark & -- & \checkmark \\
M & White & Four Year & 60+ & -- & -- & -- & -- & -- & \checkmark \\
F & Asian/Pacific Islander & Advanced & 40--49 & -- & -- & -- & \checkmark & -- & -- \\
M & Asian/Pacific Islander & Advanced & 25--39 & -- & \checkmark & -- & \checkmark & -- & -- \\
M & Asian/Pacific Islander & Advanced & 40--49 & \checkmark & \checkmark & \checkmark & \checkmark & -- & \checkmark \\
M & Asian/Pacific Islander & Advanced & 50--59 & -- & \checkmark & \checkmark & \checkmark & -- & \checkmark \\
M & Asian/Pacific Islander & Advanced & 60+ & -- & -- & \checkmark & \checkmark & -- & \checkmark \\
F & Black & Advanced & 25--39 & -- & -- & -- & \checkmark & -- & -- \\
F & Black & Advanced & 40--49 & -- & \checkmark & -- & -- & -- & -- \\
M & Black & Advanced & 25--39 & -- & \checkmark & \checkmark & \checkmark & -- & \checkmark \\
M & Black & Advanced & 40--49 & -- & \checkmark & \checkmark & \checkmark & -- & \checkmark \\
M & Black & Advanced & 50--59 & -- & -- & \checkmark & \checkmark & \checkmark & \checkmark \\
F & Other & Advanced & 25--39 & -- & -- & -- & -- & -- & \checkmark \\
F & Other & Advanced & 40--49 & -- & \checkmark & -- & \checkmark & \checkmark & \checkmark \\
F & Other & Advanced & 50--59 & -- & -- & \checkmark & \checkmark & \checkmark & \checkmark \\
M & Other & Advanced & 40--49 & \checkmark & \checkmark & \checkmark & \checkmark & \checkmark & \checkmark \\
M & Other & Advanced & 50--59 & -- & \checkmark & \checkmark & \checkmark & -- & \checkmark \\
M & Other & Advanced & 60+ & -- & \checkmark & -- & -- & -- & -- \\
F & White & Advanced & 40--49 & -- & -- & -- & \checkmark & -- & -- \\
F & White & Advanced & 50--59 & -- & \checkmark & \checkmark & \checkmark & \checkmark & \checkmark \\
M & White & Advanced & 25--39 & -- & \checkmark & \checkmark & \checkmark & \checkmark & \checkmark \\
M & White & Advanced & 40--49 & -- & -- & -- & \checkmark & \checkmark & \checkmark \\
M & White & Advanced & 50--59 & -- & -- & \checkmark & \checkmark & \checkmark & \checkmark \\
M & White & Advanced & 60+ & -- & -- & -- & -- & \checkmark & \checkmark \\
\specialrule{0.6pt}{0pt}{0pt}
\bottomrule
\end{tabular}
}
\caption{Results for the CPS work hours dataset: selected groups at levels $\alpha=0.01$ and $0.05$. The reference group corresponds to Female/White/Some college/25--39.}
\label{tab:work_hour}
\end{table}

\end{document}